\pdfoutput=1
\RequirePackage{snapshot}
\makeatletter
\def\snap@providesfile#1[#2]{%
  \wlog{File: #1 #2}%
  \if\expandafter\snap@graphic@test\expanded{#2}@@\@nil
    \snap@record@graphic#1\relax #2 (type ??)\@nil
  \else
    \expandafter\xdef\csname ver@#1\endcsname{#2}%
  \fi
  \endgroup
}
\makeatother

\documentclass[a4paper,11pt,twoside,
english,fabfinal]{fabsmfart}

\usepackage{fabmacromath,fababbreviations,fabalphabets,fabmaps,fabQFT,fabcolgraphs,fabanalysis}

\usepackage{soul,xparse,pifont}
\setstcolor{red}

\usepackage{multirow,array}

\numberwithin{equation}{section}

\usepackage[hypertexnames=false]{hyperref}
\hypersetup{%
 pdftitle={Can we make sense out of ``Tensor Field Theory''?},%
  pdfauthor={Vincent Rivasseau, Fabien Vignes-Tourneret},%
  pdfsubject={},%
  pdfkeywords={},%
 }
\usepackage[level=2]{bookmark}

\usepackage{graphbox,color,pbox}
\graphicspath{{figures/}{./}}
\usepackage{tikz}
\usetikzlibrary{positioning,matrix}
\usetikzlibrary{arrows,positioning,calc,decorations.pathmorphing}
\tikzset{
  vertex/.style={circle, draw=white,
    fill=black!50, thick, inner sep=0pt, minimum size=1.3mm},
  vertexbb/.style={circle, draw=black,
    fill=black!40, thin, inner sep=0pt, minimum size=1.1mm,
    opacity=1},
  vertexbw/.style={circle, draw=white,
    fill=black!30, semithick, inner sep=0pt, minimum size=1.3mm},
  lab/.style={circle,fill=white,outer sep=2pt,inner
    sep=0pt,auto,midway,scale=.6},
  lab2/.style={circle,fill=white,outer sep=4pt,inner sep=0pt,midway,auto,scale=.6},
  arete/.style={
    ultra thick},
  presimp/.style={white, thick},
  simp/.style={gray, thin, opacity=1},
  behind/.style={gray, thick, opacity=.3},
  prep/.style={line width=3pt, white},
  main/.style={decoration={random
      steps,segment length=5pt,amplitude=1pt}},
  maing/.style={decoration={random
    steps,segment length=6pt,amplitude=2pt}},
}

\usepackage{caption}
\DeclareCaptionStyle{smf}[justification=centering]{width=.85\textwidth,format=plain,labelformat=simple,labelsep=period,justification=justified,font={normalfont,small},labelfont={normalfont,bf,small},textfont+={up},singlelinecheck=true}
\usepackage{subcaption}

\usepackage[capitalise,nameinlink]{cleveref}

\usepackage[backend=bibtex, natbib, backref=false, style=alphabetic,giveninits=true]{biblatex} 
\usepackage{enumitem}
\setlist{leftmargin=*,labelindent=17.0pt}

\allowdisplaybreaks[1]

\usepackage{commath}

\DeclareMathOperator{\fabRe}{Re}
\renewcommand{\Re}{\fabRe}
\DeclareMathOperator{\fabIm}{Im}
\renewcommand{\Im}{\fabIm}
\newcommand{\sgn}{\textup{sgn}}

\newcommand{\End}{\textup{\textup{End}}}
\newcommand{\EndMp}{\End(M,p)}
\newcommand{\EndCO}{\End(\C,0)}

\NewDocumentCommand{\scalprodtens}{O{\ } O{\ }}{\langle #1,#2\rangle}

\newcommand{\tens}[1]{\mathbb{#1}}
\newcommand{\tensG}{\tens{G}}

\NewDocumentCommand{\ddeg}{s}{%
  \IfBooleanTF #1
  {\omega}
  {\operatorname{\omega}}
}

\newcommand{\IFR}[1]{\map{#1}}
\newcommand{\ifG}{\IFR{G}}

\newcommand{\connected}{G} 
\newcommand{\vertex}{\Gamma} 
\newcommand{\selfnrjsymbol}{\Sigma}
\NewDocumentCommand{\cof}{m O{b} O{mel}}{
  \IfNoValueTF{#2}{
    \connected^{\textup{#3}}_{\!#1}
  }{
    \connected^{\textup{#3}}_{\!#1,\textit{#2}}
  }  
}
\NewDocumentCommand{\cofslice}{m m O{b} O{mel}}{
    \IfNoValueTF{#3}{
      \connected^{\textup{#4}}_{\!#1;#2}
    }{
      \connected^{\textup{#4}}_{\!#1,\textit{#3};#2}
    }
  }
  \newcommand{\cofmrslice}[2]{\cofslice{#1}{#2}[mr]}
\NewDocumentCommand{\cofer}{m m O{mel}}{
  \connected^{\textup{#3}}_{\!#1;#2}
}

\newcommand{\selfnrjGFb}{\bar{\mathbf{\selfnrjsymbol}}^{\textup{mel}}}
\NewDocumentCommand{\opifGFb}{m}{\bar{\mathbf{\vertex}}^{\textup{mel}}_{\!#1}}
\NewDocumentCommand{\opif}{m O{b} O{mel}}{
  \IfNoValueTF{#2}{
    \vertex^{\textup{#3}}_{\!#1}
  }{
    \vertex^{\textup{#3}}_{\!#1,\textit{#2}}
  }
}
\NewDocumentCommand{\selfnrj}{O{b} O{mel}}{
  \IfNoValueTF{#2}{
    \selfnrjsymbol^{\textup{#2}}
  }{
    \selfnrjsymbol^{\textup{#2}}_{\textit{#1}}
  }
}
\NewDocumentCommand{\selfnrjmono}{O{b} O{mel}}{
  \IfNoValueTF{#2}{
    \bar\selfnrjsymbol^{\textup{#2}}
  }{
    \bar\selfnrjsymbol^{\textup{#2}}_{\textit{#1}}
  }
}
\NewDocumentCommand{\opifer}{m m O{mel}}{
  \vertex^{\textup{#3}}_{\!#1,\textit{er};#2}
}
\NewDocumentCommand{\opifermono}{m m O{mel}}{
  \bar{\vertex}^{\textup{#3}}_{\!#1,\textit{er};#2}
}
\NewDocumentCommand{\opifmono}{m O{b} O{mel}}{
    \IfNoValueTF{#2}{
      \bar{\vertex}^{\textup{#3}}_{\!#1}
    }{
      \bar{\vertex}^{\textup{#3}}_{\!#1,\textit{#2}}
    }
  }
\NewDocumentCommand{\opifmrslice}{m m O{mel}}{
  \vertex^{\textup{#3}}_{\!#1,\textit{mr};#2}
}
\NewDocumentCommand{\selfnrjmrslice}{m O{mel}}{
  \selfnrjsymbol^{\textup{#2}}_{\textit{mr};#1}
}
\NewDocumentCommand{\selfnrjmonomrslice}{m O{mel}}{
  \bar\selfnrjsymbol^{\textup{#2}}_{\textit{mr};#1}
}
\NewDocumentCommand{\opifmonomrslice}{m m O{mel}}{
  \bar{\vertex}^{\textup{#3}}_{\!#1,\textit{mr};#2}
}
\NewDocumentCommand{\opifmonobslice}{m m O{mel}}{
  \bar{\vertex}^{\textup{#3}}_{\!#1,\textit{b};#2}
}

\newcommand{\cofbare}[1]{\cof{#1}}
\newcommand{\opifbare}[1]{\opif{#1}}
\newcommand{\opifbaremono}[1]{\opifmono{#1}} 

 \newcommand{\opifmr}[1]{\opif{#1}[mr]} 

\newcommand{\opifr}[1]{\opif{#1}[r]} 
\newcommand{\opifmonor}[1]{\opifmono{#1}[r]}

\NewDocumentCommand{\disccube}{O{d}}{\sbe{H}{$\Lambda$}_{#1}} 
\NewDocumentCommand{\pdisccube}{O{d}}{\sbe{\dot H}{$\Lambda$}_{#1}} 
\NewDocumentCommand{\contcube}{O{d}}{\sbe{\cH}{$\Lambda$}_{#1}} 
\NewDocumentCommand{\pcontcube}{O{d}}{\sbe{\dot\cH}{$\Lambda$}_{#1}} 

\newcommand{\Hilb}{\cH}

\newcommand{\Htens}{\Hilb^{\otimes}}

\newcommand{\Hop}[1][]{L(\Hilb_{#1})}
\newcommand{\Hopdirect}{\Hop^{\hspace{-1pt}\times}}

\newcommand{\transpose}{\scalebox{.6}{$T$}}
\NewDocumentCommand{\Tr}{o o o}{
  \IfNoValueTF{#3}%
  {
    \IfNoValueTF{#1}%
    {\IfNoValueTF{#2}%
      {\Trace}{\Trace\sbr{#2}}
    }%
    {\IfNoValueTF{#2}
      {\Trace_{#1}}{\Trace_{#1}\sbr{#2}}
    }
  }{
    \IfNoValueTF{#1}%
    {\IfNoValueTF{#2}%
      {\Trace}{\Trace\sbr[#3]{#2}}
    }%
    {\IfNoValueTF{#2}
      {\Trace_{#1}}{\Trace_{#1}\sbr[#3]{#2}}
    }
  }
}

\usepackage{dsfont}
\DeclareMathOperator{\Itens}{\I}

\newcommand{\tuple}[1]{\mathbf{#1}}
\newcommand{\ntup}{\tuple{n}}

\newcommand{\nbtup}{\tuple{\bar n}}
\newcommand{\ptup}{\tuple{p}}

\newcommand{\qbtup}{\tuple{\bar q}}
\newcommand{\mtup}{\tuple{m}}

\newcommand{\mbtup}{\tuple{\bar m}}

\newcommand{\Torus}{\T}

\usepackage[record]{glossaries-extra}
\usepackage{glossary-bookindex}
\usepackage{glossary-mcols}
\setglossarystyle{mcolindex}

\newglossary*{voc}{Vocabulary}
\newglossary*{constant}{Scalars}
\newglossary*{tensor}{Tensors}
\newglossary*{space}{Spaces}
\newglossary*{operatortens}{Operators on $\Htens$}
\newglossary*{operatordirect}{Operators on $\Hopdirect$}
\newglossary*{misc}{Miscellanea}
\newglossary*{graphs}{Feynman graphs}

\GlsXtrLoadResources[src={NotationAndVocabulary},type=voc,match={entrytype=index},sort={en}]
\GlsXtrLoadResources[src={NotationAndVocabulary},type=constant,match={category=constant},sort={use}]
\GlsXtrLoadResources[src={NotationAndVocabulary},type=tensor,match={category=tensor},sort={use}]
\GlsXtrLoadResources[src={NotationAndVocabulary},type=space,match={category=space},sort={use}]
\GlsXtrLoadResources[src={NotationAndVocabulary},type=operatortens,match={category=operatortens},sort={use}]
\GlsXtrLoadResources[src={NotationAndVocabulary},type=operatordirect,match={category=operatordirect},sort={use}]
\GlsXtrLoadResources[src={NotationAndVocabulary},type=misc,match={category=misc},sort={use}]
\GlsXtrLoadResources[src={NotationAndVocabulary},type=graphs,match={category=graphs},sort={use}]
\renewcommand{\transpose}{t}

\title{Can we make sense out of ``Tensor Field Theory''?}
\shorttitle{Can we make sense out of ``Tensor Field Theory''?}
\author{V.~Rivasseau \& F.~Vignes-Tourneret}
\dedicatory{}

\begin{document}
\maketitle

\begin{fabsmfabstract}
We continue the constructive program about tensor field theory through
the next natural model, namely the rank five tensor theory with
quartic \glstext{melonic} interactions and propagator inverse of the
Laplacian on $U(1)^5$. We make a first step towards its construction by establishing its power counting, identifiying the divergent graphs and performing a careful study of (a slight modification of) its RG flow. Thus we give strong evidence that this just renormalizable tensor field theory is non perturbatively asymptotically free.
\end{fabsmfabstract}

\vfil
\tableofcontents
\newpage

\section*{Introduction}
\label{intro}
\etoctoccontentsline*{section}{Introduction}{1}

Recently Hairer \cite{Hairer2014aa}  solved a series of  stochastic differential equations
such as the KPZ equation or the $\phi^4_3$ equation. An advantage of
such equations is that they are better suited to Monte Carlo
computations than functional integrals. Since then, in a systematic series of impressive articles,  Hairer and
his collaborators \cite{Bruned2018aa,Chandra2016aa,Bruned2017aa}
extended their initial programme to cover the BPHZ renormalization
\cite{Bogoliubov1957gp,Hepp1966eg,Zimmermann1969jj}. In contrast to dimensional renormalization, BPHZ renormalization is adapted to 
the program of constructive field theory. It incorporates the
multi-scale expansion, a main constructive tool \cite{Feldman1985aa}, and a more up-to-date mathematical formulation of renormalization based on Hopf algebras \cite{Connes1998lr}.

To the attentive observer, constructive field theory, namely the point of view which Hairer called \emph{static},  is rapidly merging into the  
\emph{regularity structures and corresponding models} of Hairer, which
he called  the \emph{dynamic} point of view. In the language of quantum field theory, it happens that the equations which Hairer solved were all Bosonic super-renormalizable. 
Now is time for advancing the next step: the bosonic just-renomalizable quantum field theories. 
The BPHZ  renormalization was initially designed to cover theories such as QED in dimension four, the main theory at the time. 
But a profound objection were raised, initially by Landau. Now we have a name for that obstacle: QED is \emph{not  asymptotically  free}.
 Fortunately for the future of quantum field theory, the discovery that electroweak and strong interactions  are asymptotically free
were instrumental in its ``rehabilitation'' as a fundamental theory.

A famous theorem due to Coleman states that any \emph{local Bosonic asymptotically free field theory must include non-Abelian gauge theories}.
Non-Abelian gauge theories lead to an additional severe problem: the presence of Gribov ambiguities \cite{Gribov1978aa} due to gauge fixing. The way out of these difficulties is a main reason for considering the stochastic quantization \cite{Parisi1981aa}, since in this method there are no need to fix the gauge, so no need to solve Gribov ambiguities. But it remains still a tough programme.\\

On the road to this lofty goal, we propose an intermediate step which might be worth the effort in itself. It \emph{escapes} 
Coleman's theorem by being a \emph{non-local} theory.  We  have in mind
the tensor field theory. Born in the quantum gravity craddle \cite{Ambjorn1991aa,Sasakura1991aa,Gross1992aa}, 
random tensor models extend random matrix models and 
therefore were introduced as promising candidates for an \emph{ab
  initio} quantization of gravity  in rank/dimension higher than $2$. 
However their study is less advanced since they lacked for a long time
an analog of the famous 't~Hooft $1/N$ expansion for random matrix
models. Their modern reformulation \cite{Gurau2017ab,Gurau2016ab,Gurau2012ac}
 considers \emph{unsymmetrized} random tensors\footnote{{\it
     Symmetrized} random tensors are more difficult  but melons still
   dominate again \cite{Benedetti2019aa,Carrozza2018aa}, at least in
   rank 3.}, a crucial improvement which let the large $N$ limit appear
 \cite{Gurau2011ab,Gurau2011ac,Gurau2012aa}. The limit of large matrix
 models is made of planar graphs. Surprisingly perhaps, the key to the
 $1/N$ tensors is made of a new and \emph{simpler} class  of Feynman graphs  that we called \emph{melonic}. They form the dominant graphs in this limit \cite{Bonzom2011aa,Bonzom2012ac}\footnote{In quantum field theory an early reference at those Feymnam graphs appear in \cite{Calan1982aa}}.

Random tensor models can be further divided into fully invariant models, in which both propagator and interaction are left invariant  by the symmetry (such as $U(N)^{\otimes d}$),
and \emph{non-local field theories} where the propagator is for example the ordinary Laplacien on  the torus $U(1)^{\otimes d}$ (which breaks the symmetry) 
but in which the interaction is left invariant by the symmetry. To our
own surprise, such just-renormalizable models turn out to be
asymptotically free \cite{Ben-Geloun2016aa,Rivasseau2015aa}. In particular the simplest such model in this category, nicknamed $T^4_5$ theory is asymptotically free! It made them an ideal playground for advancing the mathematics both in the static sense of constructive theory and in the sense of Hairer's stochastic quantization. This fact now many years old was perhaps overlooked by the theoretical and mathematical physics community.\\

Also the tensor methods  and models in quantum gravity that one of us
baptized the tensor track
\cite{Rivasseau2011ab,Rivasseau2012ab,Rivasseau2013ac,Rivasseau2016ad,Delporte2018aa,Delporte2020aa}
was given a big boost from an unexpected corner. Since the advent of the SYK model \cite{Kitaev2015aa,Polchinski2016aa,Madacena2016aa} it appears that 
1-dimensional quantum random tensor is even richer than
the $0$-dimensional ordinary random tensor theory \cite{Witten2016aa,Gurau2016aa,Carrozza2016aa,Klebanov2017aa,}. It is approximately reparametrization invariant (\ie conformal), includes holography and it saturates the MSS bound \cite{Maldacena2016aa}.

In fact  the real applications, as it often happens, might be
elsewhere. Today we probe reality by multiples sensors. That is, we
 represent that reality by multidimensional big arrays which are, in the mathematical sense, nothing but \emph{big tensors}.
Hence we need to develop better and more versatile algorithms to probe tensors in this limit.
  Such algorithms could benefit of the modern formulation of random tensors.
This is especially true for those separating signal to noise.
 One example is tensor PCA \cite{Richard2014aa,Ben-Arous2017aa,Ben-Arous2018aa,Zare2018aa}, which extends classical matrix PCA to tensors.
Such algorithms could be applied in a variety of domains,
 high energy physics (detection of particule trajectories), spectral imaging or videos, neuroimaging, chemometrics,
  pharmaceutics,  biometrics, social networks and many more.
In fact the analysis of big tensors form a bottleneck in such a dazzling kaleidoscope
that it is no exaggeration to say that any main progress in this field
may create a revolution in artificial intelligence.\\

Now let us come down to earth. The  \emph{tensor theory new constructive program} \cite{Rivasseau2016ab} is well advanced in the super-renormalizable case \cite{Gurau2013ac,Delepouve2014ab,}. In \cite{Delepouve2014aa} 
 the $U(1)$ rank-three model  with inverse Laplacian propagator and quartic melonic interactions, which we nickname $T^4_3$, was solved.
 In \cite{Rivasseau2016aa} the $U(1)$ rank-four model $T^4_4$ was   solved. This model looks comparable in renormalization difficulty to the ordinary $\phi^4_3$ theory, but non-locality and the graphs are more complex
hence requires several additional non-trivial arguments. The next goal
is to treat just-renormalizable asymptotically free Bosonic
$T^4_5$. In 1979, G. 't~Hooft gave a series of lectures entitled
\emph{Can we make sense out of ``Quantum Chromodynamics''?}
\cite{t-Hooft1979aa}\footnote{The title of our article refers
  obviously to this seminal work.}. He presented there arguments and strategies to
control QCD via the study of its singularities in the Borel plane. To
this aim, he had to control the flow of the coupling constant in the
complex plane. The tensor field theory $T^{4}_{5}$ is a perfect
playground for constructive purposes as its flows can be controled
precisely thanks to its simple and exponentially bounded divergent
sector. In the present paper we make a further step by connecting it, modulo certain hypotheses, to an autonomous non-linear flow of the theory of dynamical systems.

The $T^{4}_{5}$ tensor field theory is precisely defined in
\cref{sec-model}. In particular, we present the cut-offs we use and an
alternative representation of the model in terms of an intermediate
matrix field. \Cref{sec-feynman-graphs} is devoted to the three
different representations of Feynman graphs we need (tensor graphs,
coloured graphs and intermediate field maps) as well as related
concepts thereof. In \cref{sec-diverg-melon-sect} we derive the
power-counting, identify the families of divergent graphs and give the
recursive definitions of the \emph{melonic} correlation functions. For
constructive purposes, we will employ none of the bare, renormalized or
even fully effective perturbative expansions. In fact, it will be
preferable to fully mass renormalize the correlation functions but use
effective wave-functions and coupling constants. We define all these
objects in \cref{sec-pert-renormalization}. We also prove there that
effective wave-functions and coupling constants are analytic functions
of the bare coupling. The main result of
\cref{sec-pert-renormalization} is \cref{thm-asymptotic-freedom} which
consists in a non perturbative definition of the RG flow for the
coupling constant. A careful study of an approximation of this flow is
carried out in \cref{sec-holomorphic-rg-flow} using tools and concepts
from discrete and continuous holomorphic local dynamical systems. We
identify in particular ``cardioid-like'' domains of the complex plane
invariant under this modified RG flow, see
\cref{thm-complex-cubic-ODE,thm-cubic-cpx-flow-bounded,thm-complex-higher-ODE-upper-bound}.\\

Solving this $T_{5}^{4}$ model means defining its correlation
functions non perturbatively in the coupling constant $g$. More precisely
it requires to prove the existence of holomorphic functions of $g$ in
a (probably cardioid-like, with a cut on the negative real axis)
domain of the complex $g$-plane such that their Taylor expansions
coincide with the perturbative expansions of the (formal) correlation
functions of the model. Moreover these functions should very probably
be proven Borel summable.

To achieve that goal, one expresses the regularized and
renormalized correlation functions as series of analytic functions,
normally convergent in a domain the size of which is uniformly
bounded in the ultraviolet cutoff. The infinite cutoff limit is then
well-defined and analytic. These expansions consist in partial
resummations of the perturbative series, either expressed in an
intermediate field representation (this is the so-called Loop Vertex
Expansion \cite{Rivasseau2007aa,Delepouve2014aa,Rivasseau2016aa}) or
obtained from a specific change of the initial tensor field variables
(in which case it is called Loop Vertex Representation
\cite{Rivasseau2017aa,Krajewski2019aa,Krajewski2019ab}). Both
approaches have pros and cons but none of them appears totally suited
for the new challenges brought by the $T_{5}^{4}$ theory. An update of
all currently known approaches to constructive tensor field theory
seems necessary \cite{Rivasseau2018aa}.

\paragraph{Acknowledgments}

F.~V.-T. thanks É.~Fouassier and T.~Lepoutre for their
kind explanations about the Cauchy-Lipshitz theorem and other basic
facts about ODEs. He also thanks N.-V.~Dang for many interesting and useful discussions and for advising him to contact
F.~Loray. It is a pleasure for us to warmly thank F.~Loray who kindly
explained us basic (but nevertheless very useful) aspects of
holomorphic dynamics. V.~R. deeply thanks his coauthor, and all those who encourage him, 
at a critical moment, to resume his scientific activity.

\section{The model}
\label{sec-model}

\subsection{Rank $5$ tensors and free Gaussian measure}
 In this section we follow as closely as possible the notations of \cite{Rivasseau2016aa}. 
Consider a pair of conjugate rank-5 tensor fields 
\begin{equation*}
  T_{\ntup}, \bar T_{\nbtup}, \text{ with }
\ntup = (n_1,n_2,n_3,n_4,n_{5}) \in \Z^5\text{, }\nbtup = (\bar
n_1,\bar n_2,\bar n_3,\bar n_4,\bar n_{5} ) \in \Z^5.
\end{equation*}
They belong respectively to the tensor product 
$\Htens \defi \Hilb_1 \otimes \Hilb_2 \otimes\Hilb_3\otimes\Hilb_4\otimes\Hilb_5$ and to its dual,
where each $\Hilb_{i}$ is an independent copy of  $\ell_2 (\Z)=
L_2 (U(1))$, and  the colour or strand index $i$ takes values in $\set{1,2,3,4,5}$.
By Fourier transform the field $T$ can be considered also as an ordinary scalar field  
$T  (\theta_{1},\theta_{2},\theta_{3}, \theta_4, \theta_{5})$ on the five-dimensional torus $\Torus_5 = U(1)^5$ and
$\bar T(\bar \theta_{1},\bar \theta_{2},\bar \theta_{3}, 
\bar\theta_4, \bar\theta_{5})$ is simply its complex conjugate 
 \cite{Ben-Geloun2011aa,Delepouve2014aa}. The tensor index $\ntup$ can
 be thought as the \emph{momenta} associated to the positions
 $\tuple{\theta}$.

 Throughout this paper, we always use bold characters to denote tuples
 of at least two variables.\\

We introduce the normalized Gaussian measure
\begin{equation*}
d\mu_{C}(T, \bar T) \defi \left(\prod_{\ntup, \nbtup}
  \frac{dT_{\ntup} d\bar T_{\nbtup}}{2i\pi} \right) \det(C^{-1}) \ 
e^{-\sum_{\ntup,\nbtup} T_{\ntup} C_{\ntup,\nbtup}^{-1} \bar T_{\nbtup}}
\end{equation*}
where the covariance $C$ is
\begin{align}\label{eq-propa}
C_{\ntup,\nbtup}  &= \delta_{\ntup,\nbtup}\, C(\ntup),& C (\ntup) &= \frac{1}{ \ntup^{2}+ m^2},&
 \ntup^{2}&\defi n_{1}^2+n_2^2+n_3^2 +n_4^2 +n_5^2.
\end{align}
This defines the \emph{free} tensor fields as random distributions
on $\Z^{5}$, namely on the dual of rapidly decreasing sequences on
$\Z^{5}$. But as we are interested in interacting tensor fields, we need to regularise the free measure.

\subsection{Ultraviolet cutoff}
\label{UV-cutoff}

In practice we want to restrict the index $\ntup$ to lie in a finite set rather than 
$\mathbb{Z}^5$ in order to have a well-defined
proper (finite dimensional) tensor model. This restriction is an ultraviolet cutoff in quantum field theory language.

A colour-factorized ultraviolet cutoff would restrict all previous sums over tensor indices 
to lie in $[-N, N]$. However it is not well adapted to the rotation invariant propagator of \eqref{eq-propa} below,
nor very convenient for multi-slice analysis as in \cite{Gurau2014ab}. 
Therefore we introduce a rotation invariant cutoff but in contrast
with \cite{Rivasseau2016aa} it will be smooth.

Let $a,\epsilon$ be two positive numbers such that $\epsilon<a$. Let $\chi_{\epsilon}$
be a smooth positive function with support $[-\epsilon,\epsilon]$. We
denote by $\indic_{[-a,a]}$ the indicator function of $[-a,a]$. In
order to prepare for multiscale analysis (see
\cref{sec-wave-funct-renorm}), we fix an integer $M>1$ (as ratio of a
geometric progression $M^j$) and choose a large integer $\jm$. Our
ultraviolet cutoff is defined as
\begin{align*}
  \kappa_{\jm}(\ntup^{2})&\defi\kappa (M^{-2\jm}\ntup^{2}),\quad\kappa(\ntup^{2})\defi\indic_{[-a,a]}\star\chi_{\epsilon}(\ntup^{2}).\\
  \intertext{It is smooth, positive, compactly supported, and satisfies}
  \kappa_{\jm}(\ntup^{2})&=
          \begin{cases}
            0&\text{if $\ntup^{2}>(a+\epsilon)M^{2\jm}$,}\\
            1&\text{if $0\les\ntup^{2}\les (a-\epsilon)M^{2\jm}$.}
          \end{cases}
\end{align*}
It is convenient to choose $a=5/2$ and $\epsilon=3/2$ so that the UV cutoff $\kappa$ effectively restricts each colour index to lie
in $[-N,N]$ with
\begin{equation*}
N\defi \floor{(a+\epsilon)^{1/2}M^{\jm}}=2 M^{\jm}.
\end{equation*}

The normalized bare Gaussian measure with cutoff $\jm$ is
\begin{equation*}
d\mu_{C_b}(T, \bar T) \defi \del[3]{\prod_{\ntup, \nbtup}
  \frac{dT_{\ntup} d\bar T_{\nbtup}}{2i\pi}} \det(C_b^{-1}) \ 
e^{-\sum_{\ntup,\nbtup} T_{\ntup} C_{b;\ntup,\nbtup}^{-1} \bar T_{\nbtup}}
\end{equation*}
where the bare covariance $C_b$ is, up to a bare field strength parameter $Z_b$, the inverse of the Laplacian on $\Torus_5$ with momentum cutoff $\jm$ plus a bare mass term
\begin{align*}
C_{b;\, \ntup,\nbtup}  &= \delta_{\ntup\nbtup} \, \kappa_{\jm}(\ntup^{2}) C_b(\ntup),& C_b (\ntup) &= \frac{1}{Z_b} \frac{1}{ \ntup^{2}+ m_b^2},&
 \ntup^{2}&\defi n_{1}^2+n_2^2+n_3^2 +n_4^2 +n_5^2.
\end{align*}
A random tensor $T$ distributed according to the measure $\mu_{C_{b}}$
is almost surely a smooth function on $U(1)^{5}$.

\subsection{The bare model}
The generating function for the moments of the model is
\begin{equation}\label{eq-Zb}
\cZ^{\raisebox{.3\height}{\scalebox{.6}{(N)}}}_{b}(g_{b},J, \bar J)= \cN  \int e^{T\scalprod\bar J+ J\scalprod\bar T}
e^{-\frac{ g_b Z_b^2}{2} \sum_c V_c(T, \bar T)} \diff\mu_{C_b}(T, \bar T), 
\end{equation}
where the scalar product of two tensors $A\scalprod B$ means
$\sum_{\ntup}  A_{\ntup}  B_{\ntup}$, $g_b$ is the \emph{bare} coupling constant (which depends on the cutoff $N$), 
the source tensors $J$ and $\bar J$
are dual respectively to $\bar T$ and $T$ and $\cN$ is a
normalization factor. To compute correlation functions it is common to choose 
\begin{equation*}
\cN^{-1} =\int \exp\del[2]{-\tfrac{ g_b Z_b^2}{2}\sum_c V_c(T,
  \bar T)} \diff\mu_{C_b}(T, \bar T),
\end{equation*}
which is the sum of all vacuum
bare amplitudes. However following the constructive tradition, we shall limit $\cN$ to be the exponential of the (infinite) sum of the \emph{divergent} connected vacuum amplitudes. Remark the $Z_b^2$ scaling factor
multiplying $g_b$ in \eqref{eq-Zb}.\\

To make the interaction $\sum_c V_c(T, \bar T)$ in \cref{eq-Zb} explicit,
we recall first some notation. $\Tr$, $\Itens$ and $\scalprodtens$
mean respectively the trace, the identity and the scalar product on
$\Htens$. $\Itens_{c}$ is the identity on $\Hilb_{c}$, $\Tr[c]$ is the trace on $\Hilb_c$
and $\scalprodtens_{c}$ the scalar product restricted to $\Hilb_{c}$. 
The notation $\hat c$ means ``every color except $c$''. For instance, 
$\Hilb_{\hat c}$ means $\bigotimes_{c' \ne c}  \Hilb_{c'}$,
$\Itens_{\hat c}$ is the identity on the tensor product $\cH_{\hat
  c}$, $\Tr_{\hat c}$ is the partial trace over $\Hilb_{\hat c}$ and
$\scalprodtens_{\hat c}$ the scalar product restricted to $\cH_{\hat c}$.

$T$ and $\bar T$ can be considered both as vectors in $\Htens$ or as diagonal 
(in the momentum basis) operators acting on $\Htens$, with eigenvalues $T_{\ntup}$ and $\bar T_{\nbtup}$. An important quantity in \glstext{melonic} tensor models
is the partial trace $\Tr[\hat c][T \bar T][0]$, which we can also
identify with the partial product $\scalprodtens[T][\bar T]_{\hat
  c}$. It is a (in general non-diagonal) operator in $\Hilb_c$ with matrix elements in the momentum basis
\begin{equation*}
\scalprodtens[T][\bar T]_{\hat c}(n_c, \bar n_c)  = \Tr[\hat c][T \bar T][1](n_c, \bar n_c) 
= \Big[\prod_{c'\neq c}\Big(\sum_{n_{c'}, \bar n_{c'}} \delta_{n_{c'} \bar n_{c'}}\Big)\Big]T_{\ntup} \bar T_{\nbtup}.
\end{equation*}
The main new feature of tensor models compared to ordinary field theories 
is the non-local form of their interaction, which is chosen invariant
under independent unitary transformations on each color index. 
In this paper we consider only
the quartic \glstext{melonic} interaction {\cite{Delepouve2014ab}}, which is a sum over colors 
$\sum_{c=1}^5 V_c(T, \bar T) $ where
\begin{align*}
V_c(T, \bar T) &= \Tr[c][(\Tr[\hat c][T \bar T][0])^2]\\
  &=\sum_{\substack{n_{c},\bar n_{c},\\m_{c},\bar m_{c}}} \Big(\sum_{\ntup_{\hat
      c},\nbtup_{\hat c}} T_{\ntup}\bar T_{\nbtup}\,\delta_{\ntup_{\hat c} \nbtup_{\hat c}} \Big) \delta_{n_c \bar m_c} 
  \delta_{m_c \bar n_c}\Big(\sum_{\mtup_{\hat c}, \mbtup_{\hat c}} T_{\mtup}\bar T_{\mbtup}\,\delta_{\mtup_{\hat c} \mbtup_{\hat c}}  \Big)\nonumber .
\end{align*}

This model is globally symmetric under colour permutations. It is just renormalizable like ordinary $\phi^4_4$
but unlike ordinary $\phi^4_4$ it is asymptotically free and using
this crucial difference, we aim, in a future work, at making rigorous
sense of it.\\

Mainly in order to prepare for the constructive study of the
$T^{4}_{5}$ model, we present here its \gls{intermediate} {\cite{Gurau2013ac}}. We put $g_b \fide \lambda_b^2$ and decompose the five interactions $V_{c}$ in \cref{eq-Zb} by introducing five intermediate Hermitian $N\times N$ matrix\footnote{The indices of $\sigma$
cannot be bigger than the maximal value $N$ of each tensor index.}  fields $\sigma^{\transpose}_{c}$ acting on $\cH_c$ (here the superscript $\transpose$ refers to transposition) and dual to $\Tr[\hat c][T \bar T][0]$, in the following way
\begin{equation*}
  e^{-\frac{ (\lambda_bZ_b)^2}{2} V_{c}(T, \bar T) }= \int e^{i\lambda_b Z_b
    \Tr[c][\del{\Tr[\hat c][T \bar T]}\sigma^{\transpose}_{c}][1]}\diff\nu(\sigma^{\transpose}_{c}),
\end{equation*}
where $\diff*\nu$ is the usual GUE measure, that is
$\diff*\nu(\sigma^{\transpose}_{c})=\diff*\nu(\sigma_{c})$ is the normalized Gaussian independently
identically distributed measure of covariance $1$ on each coefficient of the Hermitian matrix $\sigma_{c}$. It is convenient to consider $C_{b}$ as a (diagonal) operator acting on
$\Htens$, and to define in this space the operator
\begin{equation*}
  \sigma\defi \sum_c   \sigma_{c} \otimes \Itens_{\hat c} =
\sigma_{1}\otimes\Itens_2\otimes\Itens_3\otimes\Itens_4\otimes\Itens_5+
\dotsm + \Itens_1\otimes\Itens_2\otimes\Itens_3\otimes\Itens_4\otimes\sigma_{5}.
\end{equation*}
Performing the now Gaussian integration over $T$ and $\bar T$ yields 
\begin{align}
\sbe{\cZ}{N}_{b}(g,J, \bar J)&= \cN \iint e^{T\scalprod\bar J+ J\scalprod\bar
   T} e^{i\lambda_b Z_b
    \Tr[c][\del{\Tr[\hat c][T \bar T]}\sigma^{\transpose}_{c}][1]} \diff\mu_{C_b}(T,\bar T)\prod_c \diff*\nu(\sigma_{c})\nonumber\\
 &= \cN \int e^{\scalprodtens[\bar J][R(\sigma)C_{b}J]-\Tr\log(\Itens-i\lambda_{b}Z_bC_{b}\sigma)}\diff\nu(\sigma)\label{eq-Zbsigma}
\end{align}
where $\diff*\nu(\sigma)\defi \prod_c  \diff*\nu(\sigma^c)$,
and $R$ is the \gls{resolvent} operator on $\Htens$
\begin{equation*}
R(\sigma)  \defi \frac{1}{\Itens -i \lambda_b Z_b C_b\sigma }.
\end{equation*}

\section{Feynman graphs}
\label{sec-feynman-graphs}

Perturbative expansions in quantum field theory are indexed by
graphs called \glsplural{feynmangraph}. Their properties reflect
analytical aspects of the action functional. Here we will deal with three different graphical
notions.

\subsection{Tensor graphs}
\label{sec-tensor-graphs}

The first one corresponds to the Feynman graphs of action
\eqref{eq-Zb} in which the fields are tensors of rank five. As for
random matrix models, Feynman graphs are stranded graphs (so-called
ribbon graphs in the matrix case) where each strand represents the
conservation of one tensor index. The corresponding Feynman rules are
recalled in \cref{fig-tensorgraphs} where an example of such a Feynman
graph is also given. Such graphs will be called \glspl{tensorgraph} in the
sequel and denoted by emphasized letters such as $\tensG$. A
tensor graph is open if it has a positive number of external
edges\footnote{What we call external edges are actually half-edges and
open graphs are in fact pre-graphs. But we do not insist on being so
precise with our terminology.} and closed
otherwise. An open graph with $n$ external edges is often called an
\gls{nptgraph}.
\begin{figure}[!htp]
  \centering
  \begin{subfigure}[c]{.45\linewidth}
    \centering
    \raisebox{-.5\height}{\includegraphics[scale=1]{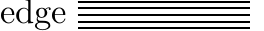}},\quad
    \raisebox{-.5\height}{\includegraphics[scale=1]{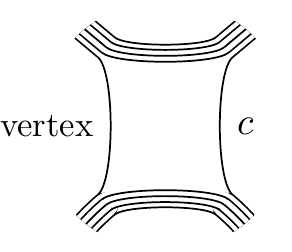}}
    \caption{Feynman rules in the tensor representation}
    \label{fig-FeynmanRulesTensorGraphs}
  \end{subfigure}\qquad
  \begin{subfigure}[c]{.5\linewidth}
    \raisebox{-.5\height}{\includegraphics[scale=1]{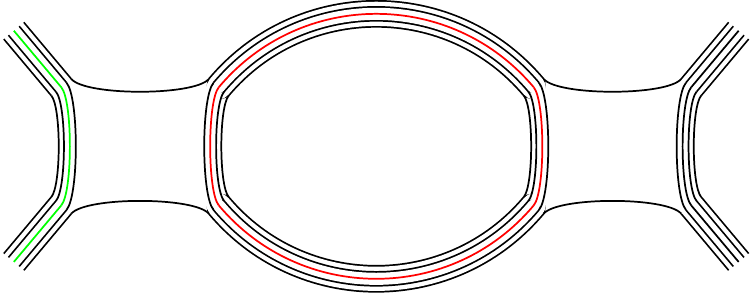}}
    \caption{An open tensor graph with $2$ vertices, $2$ internal
      and $4$ external edges. One face is drawn in red, one external path
      in green.}
    \label{fig-tensorGraph-ex}
  \end{subfigure}
  \caption{Tensor graphs}
  \label{fig-tensorgraphs}
\end{figure}
\\

The power-counting, \ie the behaviour at large $N$ of the amplitude, of
a tensor graph $\tensG$ depends on the number $F(\tensG)$ of its cycles, also called
\glspl{face}. Open tensor graphs also have non cyclic strands which we
call \glspl{extpath}, see \cref{fig-tensorGraph-ex}. It will be convenient to express the number of
faces in terms of the (reduced) \Gdegree \cite{Gurau2013ab} of
the \gls{colext} of $\tensG$. We now explain these notions.

\subsection{Coloured graphs}
\label{sec-coloured-graphs}

Strands of a tensor graph correspond to indices of the original tensor fields $T$ and $\bar
T$. Each such index is labelled by an integer from $1$ to $5$
recalling that $T$ is an element of
$\Hilb_{1}\otimes\Hilb_{2}\otimes\dotsm\otimes\Hilb_{5}$. We can then
associate bijectively to any tensor graph $\tensG$ a bipartite
$6$-regular properly edge-coloured graph $G$ called its coloured
extension. See \cref{fig-colouredgraphs} for a pictorial explanation of the bijection
as well as an example. Such edge-coloured graphs, with or without the
constraint of being $6$-regular, will be called
\glspl{colouredgraph} for simplicity and their symbols will be written
in normal font. A $(D+1)$-regular coloured graph will simply be called
$(D+1)$-coloured graph.
\begin{figure}[!htp]
  \centering
  \begin{subfigure}[c]{.45\linewidth}
    \centering
    \raisebox{-.5\height}{\includegraphics[scale=1]{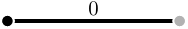}},\qquad
    \raisebox{-.5\height}{\includegraphics[scale=1]{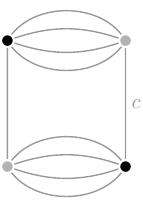}}
    \caption{Feynman rules in the coloured representation}
    \label{fig-FeynmanRulesColouredGraphs}
  \end{subfigure}\qquad
  \begin{subfigure}[c]{.5\linewidth}
    \centering
    \raisebox{-.5\height}{\includegraphics[scale=1]{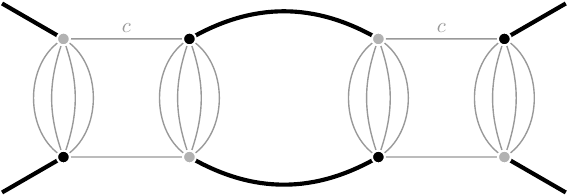}}
    \caption{The $6$-coloured graph which is in bijection with the
      tensor graph of \cref{fig-tensorGraph-ex}.}
    \label{fig-colouredGraph-ex}
  \end{subfigure}
  \caption{Coloured graphs}
  \label{fig-colouredgraphs}
\end{figure}

We will need several different notions associated to coloured
graphs. The coloured extension of a closed (\resp open) tensor graph will also be
considered closed (\resp open). In this work, edges of coloured graphs
will bear a ``colour'' in $[5]\defi\set{0,1,\dotsc,5}$. We will also
write $[5]^{*}$ for the set $\set{1,2,\dotsc,5}$. Let $G$ be a
coloured graph. We let $\gls{colset}[(G)]$ be the set of colours labelling at
least one edge of $G$. Let $c$ be an element of $[5]$. We will often
write $\hat c$ for $[5]\setminus\set{c}$. We denote by $E_{c}(G)$
the set of edges of $G$ of colour $c$. The elements of $E_{c}(G)$ are
the $c$-edges of $G$. If $\cC$ is a subset of $[5]$,
we denote $\bigcup_{c\in\cC}E_{c}(G)$ by $E_{\cC}(G)$. Let $E'$ be any subset of the
edges of $G$, we let $G[E']$ be the spanning subgraph of $G$ induced
by the edges in $E'$: the vertex-set of $G[E']$ is the same as the one
of $G$, and the edge-set of $G[E']$ is $E'$.

Certain (coloured) subgraphs of coloured graphs play a particularly important
role. We let again $\cC$ be a subset of $[5]$. A $\cC$-\gls{bubble} is a connected component of $G[E_{\cC}]$. Let $n$ be an
element of $\set{0,1,\dotsc,6}$. An $n$-bubble $B$ is a bubble such
that $\colset(B)$ has cardinality $n$. A $2$-bubble of a \emph{closed}
coloured graph is therefore a cycle whose edges bear two alternating
colours. Cyclic $2$-bubbles of $G$ whose colour set belong to
$\set{\set{0,i},\ i\in[5]^{*}}$ correspond to the faces of the
corresponding tensor graph $\tensG$. By extension, cyclic $2$-bubbles are often also called
faces and their number denoted $F(G)$. The number of
$\set{0,c}$-bubbles will be written $F_{0c}$ and we define $F_{0}(G)\defi\sum_{c\in[5]^{*}}F_{0c}$ so that $F_{0}(G)=F(\tensG)$. Similarly, we denote by $F_{\emptyset}(G)$ the total
number of faces of $G$, both colours of which are different from
$0$. Non cyclic $2$-bubbles of $G$ represent the external paths of
$\tensG$. The interaction vertices of a tensor graph $\tensG$ are in bijection with the $\hat 0$-bubbles of its coloured extension $G$.

The (reduced) \gls{Gdegree} $\rGdeg(G)$ of a \emph{closed}
$(D+1)$-coloured graph\footnote{In this case, by convention, the set of
  colours of $G$ is $[D]$.} $G$ is
defined as follows \cite{Gurau2013ab}:
\begin{equation*}
  \rGdeg(G)\defi\tfrac 14D(D-1)V(G)+D\,C(G)-F(G)
\end{equation*}
where $V(G)$ is the number of vertices of $G$ and $C(G)$ its
number of connected components. It is a
non negative integer. One can indeed show that it is the sum of the
genera of some maps associated to $G$ \cite{Gurau2012aa}. Let
$\cycperm_{[D]}$ be the set of cyclic permutations of $[D]$
and $\tau$ be such a permutation. Let $\jack_{\tau}(G)$
($\jack_{\tau}$ if the context is clear) be the
map whose underlying graph is $G$ and whose cyclic ordering of the
edges around vertices is given by $\tau$. Such maps are called \glspl{jacket}
in the tensor field literature \cite{Ben-Geloun2010aa}, see \cref{fig-jacket-ex} for an example. Then, we have
\begin{equation*}
  \rGdeg(G)=\frac{1}{(D-1)!}\sum_{\tau\in\cycperm_{[D]}}g_{\jack_{\tau}}.
\end{equation*}
\begin{figure}[!htp]
  \centering
  \raisebox{-.5\height}{\includegraphics[scale=1]{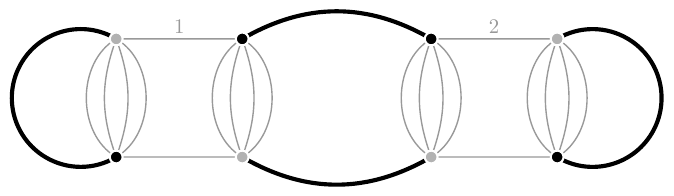}}\hspace{2cm}
  \raisebox{-.5\height}{\includegraphics[scale=1]{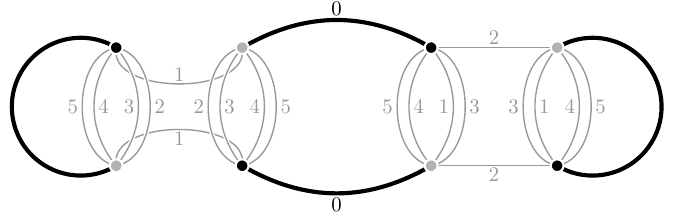}}
  \caption{The jacket $\jack_{(023145)}$ (right) of the left coloured graph}
  \label{fig-jacket-ex}
\end{figure}

In order to classify the divergent graphs of tensor models, one needs
an extension of the \Gdegree to \emph{open} coloured graphs and
the notion of boundary graph of a coloured graph. To start with, we
need a slightly generalised version of jacket\footnote{Such a map is
  called a pinched jacket in \cite{Ben-Geloun2011aa}.}.
\begin{defn}[Jacket of a possibly open coloured graph]
  \label{def-generalized-jacket}
  Let $G$ be a $(D+1)$-coloured graph, open or closed. Let $\tau$ be a
  cyclic permutation of $[D]$. The jacket $\jack_{\tau}$ of
  $G$, with respect to $\tau$, is the map built in the following way:
  \begin{enumerate}
  \item consider $G$ as a graph,
  \item fix the cyclic ordering of its edges around its vertices
    according to $\tau$,
  \item delete the half- (or external) edges of this map. 
  \end{enumerate}
\end{defn}
Before defining a natural version of the \Gdegree for possibly open
coloured graphs, we need to remind the reader of the notion of a
\gls{bdrygraph}. It is well-known that open (\resp closed) $(D+1)$-coloured graphs
encode (triangulated) $D$-dimensional piecewise linear normal pseudo-manifolds with
(\resp without) boundary \cite{Pezzana1974aa,Gurau2012ab}. Let $G$ be
such an open coloured graph. The boundary of its dual pseudo-manifolds is triangulated by a
complex dual to the boundary graph $\gls{dG}$ of $G$. $\dG$ is a
$D$-coloured graph defined as follows: its vertex-set is the
set of external edges of $G$. Its edge-set is the set of non cyclic
$2$-bubbles of $G$. \Cref{fig-boundary-ex} provides examples of
boundary graphs. From the bijection between coloured graphs and tensor
graphs, $\partial G$ defines the boundary (tensor) graph
$\partial\tensG$ of $\tensG$.
\begin{figure}[!htp]
  \centering
  \begin{tikzpicture}
    \def\u{6};
    \node (un) at (0,0) [shape=rectangle]{\includegraphics[scale=.8]{fondamental4}};
    \node (dun) at (\u,0)
    [shape=rectangle]{\includegraphics[scale=.8]{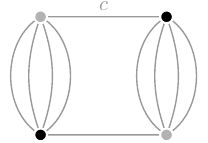}};
    \node (deux) at (0,-.4*\u) [shape=rectangle]{\includegraphics[scale=.8]{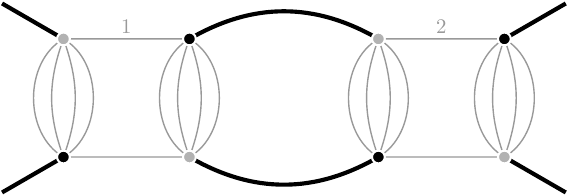}};
    \node (ddeux) at (\u,-.4*\u)
    [shape=rectangle]{\includegraphics[scale=.8]{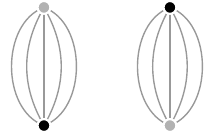}};

    \draw [thick,->] (un.east)--(dun.west) node [above,text
    width=3cm,text centered,midway] {$\partial$};
    \draw [thick,->] (deux.east)--(ddeux.west) node [above,text width=3cm,text centered,midway] {$\partial$};
  \end{tikzpicture}
  \caption{Boundary graphs}
  \label{fig-boundary-ex}
\end{figure}

\begin{defn}[Gurau degree of a possibly open coloured graph]
  \label{def-generalized-Gdegree}
  The \Gdegree of a possibly open $(D+1)$-coloured graph is defined as
  \begin{equation*}
    \rGdeg(G)\defi \frac{1}{(D-1)!}\Big(\sum_{\tau\in\cycperm_{[D]}}g_{\jack_{\tau}(G)}-\sum_{\tau\in\cycperm_{[D]^{*}}}g_{\jack_{\tau}(\dG)}\Big).
\end{equation*}
\end{defn}
If $G$ is closed, $\dG$ is empty and this equation reduces to the
usual \Gdegree for closed graphs. This definition essentially
originates from \cite{Ben-Geloun2011aa}. The point of interest for us is that, in the
case where all the external edges of $G$ bear the same colour, it can
be shown that $\rGdeg(G)$ is a measure of the number of faces of $G$:
\begin{lemma}\label{thm-Gdegree-measures-faces}
  Let $G$ be a possibly open $(D+1)$-coloured graph. Its \Gdegree,
  given by \cref{def-generalized-Gdegree}, follows
  \begin{equation*}
    \rGdeg(G)=\tfrac 14D(D-1)V(G)+D\,C(G)-F(G)-C(\dG)-\tfrac 12(D-1)E(G)
  \end{equation*}
  where $E(G)$ is the number of external edges of $G$.
\end{lemma}
\begin{proof}
  Let $\tau$ be a cyclic permutation of $[D]$ and $\jack_{\tau}(G)$ be
  the corresponding jacket of $G$. By Euler relation,
  \begin{equation*}
    g_{\jack_{\tau}(G)}=\tfrac 12\del{2C(G)-F(\jack_{\tau})+e(G)-V(G)}
  \end{equation*}
  where $e(G)$ denotes the number of internal edges of $G$. Faces of
  $\jack_{\tau}$ can be divided into two parts: the ones which are
  faces of $G$, the other ones which are not. The latter will be
  called external and their number will be denoted by
  $F_{\text{ext}}(\jack_{\tau})$. What are these external faces
  exactly? Let us look at the upper left graph of
  \cref{fig-boundary-ex}, considered as a map. Recall that in the
  definition of a jacket, we removed external edges. If one does so
  for this map, it will contain an external face which goes all around
  it. This face is not a face of $G$ because it is bordered by edges
  of three different colours, $0$, $i$ and $j$ with $i,j\neq 0$. It
  will be convenient to define external faces of $G$. Non-cyclic
  $2$-bubbles $b$ of $G$ are bordered by two external vertices, namely
  its two vertices incident with external edges. If the edges of $b$ bear
  colours $0$ and $i$, we call $b$ an external path of $G$ of colour $i$. An
  external face of $G$ of colour $ij$ is then defined as a cyclic and
  alternating sequence of adjacent external paths of colour $i$ and
  $j$ respectively. The important point to notice is that the external
  faces of $G$ are in bijection with the faces of $\partial G$.

  External faces of $G$ of colour $ij$ are faces of $\jack_{\tau}$ \ifft $\tau$ contains the sequence
  $i0j$ or $j0i$. Then, a given external face of $G$ belongs to
  exactly $2(D-2)!$ jackets. Each face of $G$ belongs to exactly $2(D-1)!$ jackets so that
  \begin{align*}
    \frac{1}{(D-1)!}\sum_{\tau\in\cycperm_{[D]}}g_{\jack_{\tau}(G)}&=\tfrac
                                                                  1{2(D-1)!}\big(2(D!)C(G)-2(D-1)!F(G)-2(D-2)!F_{\text{ext}}(G)\\
    &\phantom{=}+D!\del{e(G)-V(G)}\big)\\
                                                &=D\,C(G)-F(G)-\tfrac 1{D-1}F_{\text{ext}}(G)+\tfrac{D(D-1)}{4}V(G)-\tfrac{D}{4}E(G),
  \end{align*}
  where we used that the total number of jackets of $G$ is $D!$ and
  $2e(G)+E(G)=(D+1)V(G)$. Similarly, using that the total number of
  jackets of $\partial G$ is $(D-1)!$, that $2e(\partial
  G)=DV(\partial G)$ and that $V(\partial G)=E(G)$, we have
    \begin{align*}
    \frac{1}{(D-1)!}\sum_{\tau\in\cycperm_{[D]^{*}}}g_{\jack_{\tau}(\partial
      G)}&=\tfrac
                                                                  1{2(D-1)!}\big(2(D-1)!C(\partial
           G)-2(D-2)!F_{\text{ext}}(G)\\
      &\phantom{=}+(D-1)!\del{e(\partial G)-V(\partial G)}\big)\\
                                                &=C(\partial G)-\tfrac
                                                  1{D-1}F_{\text{ext}}(G)+\tfrac{D-2}4
                                                  E(G).
    \end{align*}
    This concludes the proof.
\end{proof}
Coloured graphs of vanishing degree are said to be \gls{melonic}. They form the
dominant family of the $1/N$-expansion of coloured tensor models
\cite{Gurau2012aa}.

\subsection{Intermediate field maps}
\label{sec-interm-field-maps}

The third graphical notion we will deal with corresponds to the
Feynman graphs of action \eqref{eq-Zbsigma} \viz Feynman graphs of the
\ifrt of our model. As the \ifrt is a multi-matrix model, its Feynman
graphs are ribbon graphs or maps. As each field $\sigma_{c}$ bears a
colour index $c$ (and the covariance is diagonal in this colour
space), the edges of these maps bear a colour too. The $\Tr\log$
interaction term implies that there is no constraint on the degrees of
the vertices of these maps nor on the properness of their
edge-colouring. Such maps will be called \glspl{colmap}. As for
coloured graphs, we let $\colset(\ifG)$ be the set of colours
labelling at least one edge of $\ifG$.

There is a bijection between the Feynman maps of the \ifrt and the
Feynman graphs of the original tensorial action \eqref{eq-Zb}. A
precise description of this bijection can be found in \cite{Bonzom2015ab}. Let us
remind the reader of its most salient features. Firstly, note that we
will in fact explain a bijection between coloured graphs and coloured
maps. Let $G$ be a 
$6$-coloured graph of the $T^{4}_{5}$ model and let
$\gls{ifG}$ be the corresponding coloured map. In each $\hat
0$-bubble of $G$, there are two sets of four parallel edges. Each set
will be called a \gls{partnerlink}.

Edges of $\ifG$ are in
bijection with the $\hat 0$-bubbles (or interaction vertices) of
$G$. Each such bubble has a distinguished colour, namely the colour
common to the two edges which do not belong to a partner link. We
label the corresponding edge of $\ifG$ with it. Partner links of $G$ are in
bijection with half-edges of $\ifG$. Let us now describe the vertices of $\ifG$. They form cycles of
half-edges. But there is a subtlety due to external edges of $G$. Each maximal alternating sequence of adjacent
$0$-edges and partner links in $G$ form either a cycle or a (non cyclic) path
in case of external ($0$-)edges. In any case, we represent such a
sequence as a vertex in $\ifG$. If a sequence is not cyclic, we add a
cilium, \ie a mark, to the corresponding vertex of $\ifG$. See
\cref{fig-bij-colgraphs-maps} for
an illustration of this bijection.
\begin{figure}[!htp]
  \centering
  \begin{tikzpicture}
    \matrix[row sep={.5cm}, column sep=3cm]{
      \node (ul)
      {\includegraphics[scale=.7,align=c]{QuarticVertex90}};&\node
      (ur) {\includegraphics[scale=1,align=c]{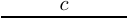}};\\
      \node (ml)
      {\includegraphics[scale=.7,align=c]{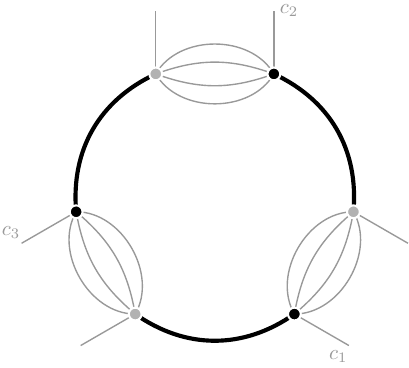}};&\node (mr) {\includegraphics[scale=1,align=c]{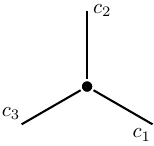}};\\
      \node (bl)
      {\includegraphics[scale=.7,align=c]{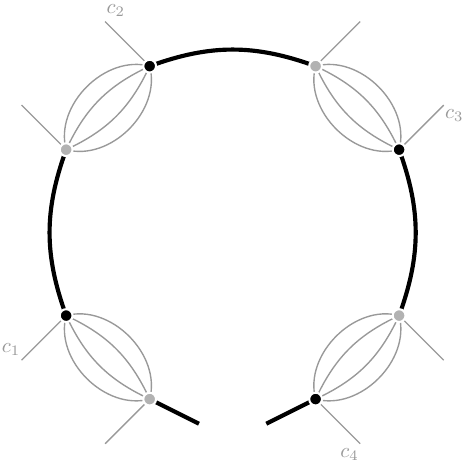}};&\node
      (br) {\includegraphics[scale=1,align=c]{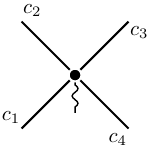}};\\
  };
  \draw[->, shorten >=.5cm, shorten <=.5cm] (ml) -- (mr);
  \draw[->, shorten >=.5cm, shorten <=.5cm] (ul) -- (ur);
  \draw[->, shorten >=.5cm, shorten <=.5cm] (bl) -- (br);
  \end{tikzpicture}
  \caption{Bijection between coloured graphs and coloured maps --
    Feynman rules}
  \label{fig-bij-colgraphs-maps}
\end{figure}

In the sequel, we will use (at least) two features of this bijection
between coloured graphs of our model and coloured maps:
\begin{enumerate}
\item $0$-edges correspond to corners of the coloured map \cite{Bonzom2015ab},
\item melonic coloured graphs are in bijection with coloured plane
  trees \cite{Dartois2013aa,Bonzom2015ab}.
\end{enumerate}

\section{The divergent melonic sector}
\label{sec-diverg-melon-sect}

\subsection{Divergent graphs}
In order to exploit the \emph{divergence} degree 
\begin{equation*}
\ddeg(\tensG)=-2L(\tensG)+F(\tensG)
\end{equation*}
of a graph $\tensG$, where $L(\tensG)$ and $F(\tensG)$ are respectively the
 number of edges and faces of $\tensG$, we need to compute the number of
 its faces. This quantity is given by
\begin{equation}
F(\tensG)=4(V(\tensG) + 1)-2E(\tensG) -\del{C(\partial\tensG)-1}-\rGdeg(G), \label{eq-tensFaces}
\end{equation}
where $V(\tensG)$ and $E(\tensG)$  are respectively the number of vertices and the
number of external legs  of $\tensG$. Indeed, as
$F(\tensG)=F_{0}(G)=F(G)-F_{\emptyset}(G)$, we can use
\cref{thm-Gdegree-measures-faces} and the specific form of the quartic
melonic interaction vertices (to compute
$F_{\emptyset}(G)=(D-1)^{2}V(\tensG)$). To get \cref{eq-tensFaces} we
also used $V(G)=4V(\tensG)$ and set $D=5$. The original proof of
\cref{eq-tensFaces} can be found in \cite{Ousmane-Samary2012ab}.

After substituting the combinatorial relation $2L+E=4V$, the divergence degree of $\tensG$ can be written as
\begin{equation}\label{eq-divdegree}
\omega(\tensG)= 4-E(\tensG) - \del{C(\partial\tensG)-1} - \rGdeg(G).
\end{equation}

\begin{lemma}[Superficially divergent graphs]\label{thm-divgraphs}
  The superficially divergent graphs, \ie the graphs $\tensG$ such that
  $\ddeg(\tensG)\ges 0$, all belong to one of the cases listed in
  \cref{tab-supdivgraphs}. Moreover, in the intermediate field representation, 
  \begin{itemize}[leftmargin=*]
  \item divergent four-point graphs are trees such that the unique
    path between their two cilia is monochrome,
  \item the closed superficially divergent
    graphs are
    \begin{itemize}
      \item plane trees if $\ddeg*=5$,
      \item unicyclic maps if $\ddeg*=0$ or $\ddeg*=2$
    \end{itemize}
  \end{itemize}
  Finally, in the latter case, $\ddeg(\tensG)=2$ \ifft the unique
  cycle of $\ifG$ is monochrome.
  \begin{table}[!htp]
  \begin{displaymath}
    \begin{array}{lcccc||}
      E(\tensG)& C(\partial\tensG)&\rGdeg(G)& \omega(\tensG) \\
      \hline\hline
      4 & \multirow{2}{*}{1} & \multirow{2}{*}{0} & 0 \\
      \cline{1-1}\cline{4-4}
      2 &  &  & 2\\
      \hline
      \multirow{3}{*}{0} & \multirow{3}{*}{0} & 0 & 5\\
      \cline{3-4}
      && 3 & 2\\
      \cline{3-4}
      && 5 & 0\\
      \hline
      \hline
    \end{array}
  \end{displaymath}
  \caption{Characteristics of superficially divergent graphs}
  \label{tab-supdivgraphs}
\end{table}
\end{lemma}
\begin{proof}
  \fabtextcite{Gurau2013ab} defined two very
  convenient coloured graphs we will need. The first one is a chain. Chains
  can be broken or unbroken. In our case, chains of an
  intermediate field map $\ifG$ are
  paths of the form $(e_{1},v_{1},e_{2},v_{2},\dotsc,e_{n-1},v_{n-1},e_{n})$ where the
  $e_{i}$'s are edges of $\ifG$, the $v_{i}$'s are vertices of $\ifG$
  such that for all $i$ between $1$ and $n-1$, the degree of $v_{i}$
  in $\ifG$ is two. Such a chain is unbroken if all its edges bear the
  same colour. It is broken otherwise. The second simple but very
  useful object is that of trivial coloured graphs or ring graphs. They
  consist in a single loop and no vertex. This loop bears a colour.
  In our case, this will always be the colour $0$. Ring graphs are
  melonic by convention and are represented by an isolated vertex in
  the \ifrt.

  According to \cref{eq-divdegree}, the divergence degree of a $4$-point
  graph $\tensG$ is bounded above by zero. It vanishes \ifft
  $C(\partial\tensG)=1$ and $\rGdeg(G)=0$. Divergent four-point graphs
  are thus trees with two cilia in the \ifrt. Now, recursively remove
  all degree one vertices of this tree which do not bear a cilium. One
  gets a non trivial path $\IFR{P}$ with a cilium at each end. This path has the
  same power counting as the initial tree $\ifG$. It is melonic and its
  boundary graph is connected if $\IFR P$ is monochrome, disconnected
  otherwise. Thus, according to \cref{eq-divdegree}, $\ifG$ is
  superficially divergent \ifft the unique path between its two cilia
  is monochrome.
  
  Let us now consider a
  Feynman graph $\tensG$ such that $E(\tensG)=2$. The
  divergence degree of such a graph is bounded above by two. The coloured
  extension of its boundary graph is the unique $6$-coloured graph
  with two vertices. It is thus connected \ie $C(\partial\tensG)=1$. Then
  $\ddeg(\tensG)=2\text{ \ifft}\rGdeg(G)=0$. Moreover, as proven in
  \cite{Ben-Geloun2012ab}, if $\rGdeg(G)>0$ then $\rGdeg(G)\ges D-2$
  where $D+1$ is the number of colours of $G$. In our case, $D$ equals
  five and the smallest possible positive degree is
  three. Consequently the only superficially divergent $2$-point
  graphs have vanishing degree.

  Let us finally treat the case of a closed ($E=0$) superficially
  divergent Feynman graph and work in the \ifrt. Note that the
  divergence degree of such a graph is bounded above by five. As a
  consequence, it has excess at most one. Indeed, adding an edge to a
  connected graph $\ifG$ increases the number of its corners by two
  (hence the number of edges of $\tensG$ increases by two) while the
  number of faces of $\tensG$ can at most increase by one. Thus the
  divergence degree decreases by at least three. A connected closed graph
  $\tensG$ with maximal divergence degree (five) is melonic and
  corresponds, in the \ifrt, to a tree. According to the argument
  above, a superficially divergent closed graph has an excess smaller
  or equal to one.

  Let us focus on divergent graphs $\ifG$ of excess one. They are maps
  with exactly two faces \ie maps with a unique cycle and trees
  attached to the vertices of this cycle. In order to further classify
  such divergent graphs, as in \cite{Gurau2013ab}, we
  first remove recursively all vertices of degree one. This does not change the
  degree of the graph. The result is a cycle $\IFR C$ \ie a ring graph
  into which a maximal proper chain $\IFR{Ch}$ has been inserted. According to
  \cite[p.\ $288$]{Gurau2013ab} the Gurau degree of the coloured graph $C$ is $3$
  if $\IFR C$ is monochrome (the chain $\IFR{Ch}$ is then non-separating and
  unbroken with a single resulting face). It is $5$ otherwise
  ($\IFR{Ch}$ is then a non-separating broken chain). 
\end{proof}

The $T^{4}_{5}$ model \eqref{eq-Zb} has the power counting of a just
renormalizable theory (and can be proven indeed perturbatively
renormalizable by standard methods). However the structure of
divergent subgraphs is simpler both than in ordinary $\phi^4_4$ or in
the Grosse-Wulkenhaar model \cite{GrWu04-3} and its translation-invariant renormalizable version.

Melonic graphs with zero, two and four external legs are divergent, respectively as $N^5$,
$N^2$ and $\log N$. In the sequel we will only consider 1PI (\ie one
particle-irreducible or $2$-edge connected) graphs as they represent the
only necessary renormalizations. Melonic graphs are trees in the
\ifrt. The condition that they are 1PI exactly corresponds to the ciliated vertices
being of degree one in the tree (cilia do not count). Melonic vacuum
graphs are always 1PI.

The divergent melonic graphs of the theory are obtained respectively
from the fundamental melonic graphs of \cref{fig-melonicdivergences}, by
recursively inserting the fundamental 2-point melon on any bold line, or, in the case of the four-point function,
also replacing any interaction vertex by the fundamental
4-point melon so as to create a ``melonic chain'' of arbitrary length
(see \cref{fig-chain-ex} for a chain of length two), in which all vertices must be of the same colour (otherwise the graph won't be divergent).
\begin{figure}[!htp]
  \begin{center}
    \includegraphics[align=c, scale=.8]{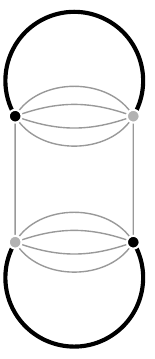} \hspace{1.5cm} \includegraphics[align=c, scale=.8]{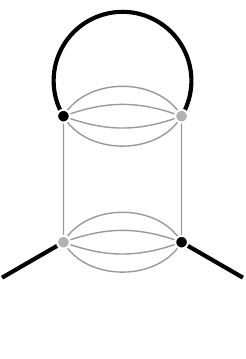}\hspace{1.5cm}
    \includegraphics[align=c, scale=.8]{fondamental4}
  \end{center}
  \caption{From left to right, the fundamental melons for the 0-, 2- and 4-point function.}
\label{fig-melonicdivergences}
\end{figure}
\begin{figure}[!htp]
  \begin{center}
    \includegraphics[scale=.8]{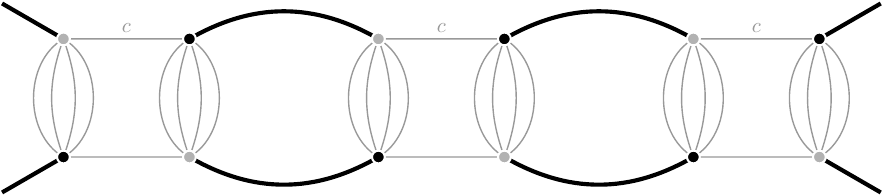}
  \end{center}
  \caption{The length-two melonic four-point chain}
  \label{fig-chain-ex}
\end{figure}

Beyond melonic approximation there is only one simple infinite family of non melonic graphs who are divergent. They are vacuum graphs diverging  either as $N^2$ or as $\log N$.
They are made of a ``necklace chain'' of arbitrary length $p\ges 1$, decorated with arbitrary
$2$-point melonic insertions. Two such necklace chains, of length one and four,
are pictured in \cref{fig-vacuumnonmelonicdivergences}. If all
couplings along the chains have same colour, the divergence is quadratic, in $N^2$. If some couplings are different, the divergence is logarithmic, in $\log N$.
\begin{figure}[!htp]
  \begin{center}
    \includegraphics[align=c]{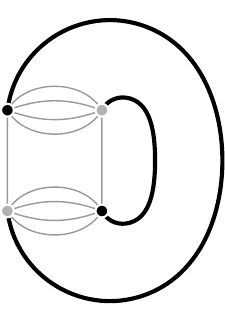}\hspace{2.5cm}\includegraphics[align=c]{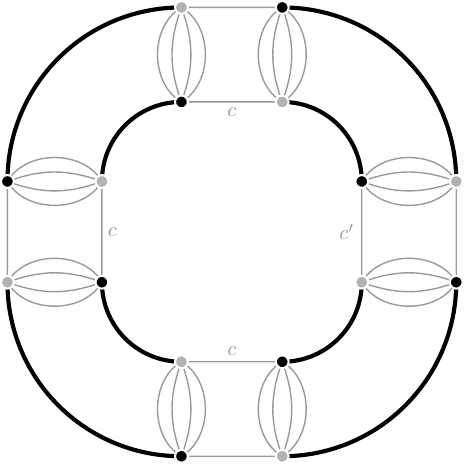}
  \end{center}
  \caption{A length-one and a length-four non-melonic divergent vacuum connected necklace. Remark that the left necklace diverges a $N^2$,
    whereas  the right one diverges as $\log N$ if $c' \ne c$.}
  \label{fig-vacuumnonmelonicdivergences}
\end{figure}

\subsection{Melonic correlation functions}
\label{sec-melon-corr-funct}

Let us call $\cof{E}$ and $\opif{E}$ respectively the connected and one-particle irreducible melonic functions (\ie sum over the melonic
Feynman amplitudes) of the theory  with $E$ external fields. With a
slight abuse of notation, the bare melonic two-point function $\cofbare{2}(\ntup, \nbtup) = \delta_{\ntup,\nbtup}\cofbare{2}(\ntup)$ is related to the bare melonic self-energy $\selfnrj(\ntup, \nbtup) = \delta_{\ntup,\nbtup}\selfnrj(\ntup)$ by the usual equation
\begin{equation*}
  \cofbare{2} (\ntup) =  \frac{C_b(\ntup)}{1 -  C_b(\ntup) \selfnrj(\ntup)}.
\end{equation*}
$\selfnrj(\ntup)$ is the sum over colors $c$ of a unique (monochrome)
function $\selfnrjmono$ of the single integer $n_c$:
\begin{equation*}
  \selfnrj(\ntup) = \sum_c  \selfnrjmono(n_c).
\end{equation*}
$\selfnrj$ is uniquely defined by the last two equations and the
following one (see \cref{fig-relationship-G2Gamma2})
\begin{equation}
  \selfnrjmono(n_c ) = - g_b Z_b^2  \sum_{\tuple p\in\Z^{5}} \delta_{p_c,n_c}\cofbare
  2(\tuple p)=- g_b Z_b^2  \sum_{\tuple
    p\in\Z^{5}}\frac{\delta_{p_c,n_c}}{C_{b}^{-1}(\tuple p)- \selfnrj(\tuple p)}.\label{eq-defGamma2}
\end{equation}
\begin{figure}[!htp]
  \centering
  \raisebox{-.5\height}{\includegraphics[scale=.6]{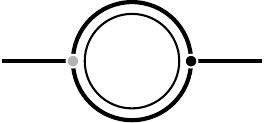}}\qquad$=$\qquad
  \raisebox{-.3\height}{\includegraphics[scale=.8]{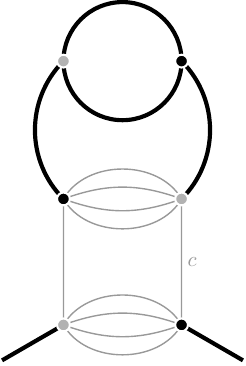}}
  \caption{Pictorial representation of the relationship between
    $\cofbare 2$ and $\selfnrjmono$. A circle stands for a connected
    function, two concentric ones for a 1PI monochrome function.}
  \label{fig-relationship-G2Gamma2}
\end{figure}

Similarly the \emph{bare} melonic four-point vertex function
$\opifbare 4(\ntup,\nbtup, \mtup, \mbtup)$ is the sum over colors $c$ of contributions
defined through a unique matrix $\opifbaremono 4(n_c, \bar n _c)$ which corresponds to the melonic invariant $V_c$:
\begin{equation*}
  \opifbare 4(\ntup,\nbtup, \mtup, \mbtup) = \sum_{c=1}^{5} \delta_{\ntup_{\hat c}  , \nbtup_{\hat c}} 
  \delta_{\mtup_{\hat c}  , \mbtup_{\hat c}}   \delta_{n_{c}  , \bar m_{c}}  \delta_{m_{c}  , \bar n_{c}} 
  \opifbaremono 4 (n_c,\bar n_c).
\end{equation*}
$\opifbare 4$ is uniquely defined by the previous equation and the
following one (see \cref{fig-relationship-G2Gamma4mono})
\begin{equation*}
  \opifbaremono 4 (n_c ,\bar n_c) = - g_bZ_b^2 \sbr[3]{1  +\sum_{\ptup,
    \qbtup\in\Z^{5}}\delta_{\ptup_{\hat c},\qbtup_{\hat
      c}}\delta_{p_c,n_c}  \delta_{\bar q_c, \bar n_c} \cofbare
  2(\ptup)  \cofbare 2(\qbtup) \opifbaremono 4(n_c ,\bar n_c)},
\end{equation*}
which solves to
\begin{equation}\label{eq-Gammabar4-expression-fct-G2}
  \opifbaremono 4 (n_c ,\bar n _c) = \frac{ -g_bZ_b^2}{1 + g_b  Z_b^2\sum_{\ptup, \qbtup}\delta_{\ptup_{\hat c},\qbtup_{\hat c}}\delta_{p_c,n_c}  \delta_{\bar q_c, \bar n_c} \cofbare 2(\ptup)  \cofbare 2(\qbtup)}.
\end{equation}
\begin{figure}[!htp]
  \centering
  \raisebox{-.5\height}{\includegraphics[scale=.8]{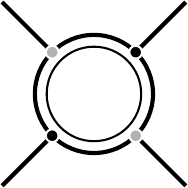}}\qquad$=$\qquad
  \raisebox{-.5\height}{\includegraphics[scale=.8]{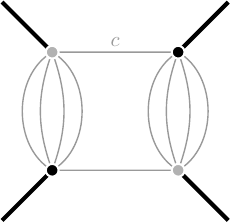}}\quad + \quad
  \raisebox{-.5\height}{\includegraphics[scale=.8]{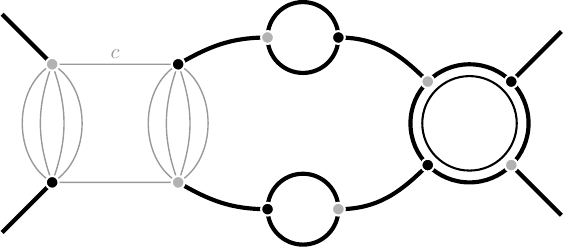}}
  \caption{Pictorial representation of the relationship between
    $\opifbaremono 4$ and $\cofbare 2$ }
  \label{fig-relationship-G2Gamma4mono}
\end{figure}

At fixed cutoff $N=2M^{\jm}$, these equations define $\selfnrj$, $\cofbare 2$
and $\opifbare 4$ (hence also $\cofbare 4$) at least as analytic functions for $g_bZ_b^{2}$ sufficiently small, 
because the species of melonic graphs is exponentially bounded as the
number of vertices increases, see \cref{sec-analyticity}. However this does not allow to take the limit $N\to \infty$ since the radius of convergence shrinks to zero in this limit. In short we need to now renormalize.

\section{Perturbative renormalization}
\label{sec-pert-renormalization}

\subsection{Renormalized 1PI functions}

The renormalization consists in a melonic BPHZ scheme which is given by BPHZ-like normalization conditions at zero external momenta, but restricted to the divergent sector, namely melonic graphs\footnote{The true BPHZ prescription in standard field theory imposes conditions on the full 1PI functions of the theory, not just their melonic part. This is because all 1PI graphs diverge in standard field theory. In this tensorial theory since non-melonic graphs are convergent the full BPHZ prescription is not minimal, and differs from the melonic BPHZ prescription only by unnecessary finite renormalizations.}.

The standard renormalization procedure expresses the 1PI correlation
functions in terms of renormalized quantities through a Taylor
expansion:
\begin{align*}
  \opifbare 2&\defi C_{b}^{-1}-\selfnrj,\\
  \opifbare 2(\ntup)&=\opifbare 2(0)+\opifmr 2(\ntup)=\opifbare 2(0)+ \ntup^2\frac{\partial\opifbare 2}{\partial \ntup^2}(0)+
  \opifr 2 (\ntup),\\
  \opifbaremono 4(n_{c},\bar n_{c})&=\opifbaremono 4 (0,0)+\opifmonor
  4(n_{c},\bar n_{c})
\end{align*}
where the subscript $\textit{mr}$ means mass-renormalized, and
$\textit{r}$ means full renormalization, with
\begin{equation*}
  m_{r}^{2}=Z_b m_b^2-\selfnrj(0),\quad 1=Z_{r}=Z_{b}-\frac{\partial \selfnrjmono}{\partial n_{c}^2}(0), \quad \opifbaremono 4 (0,0)=-g_{r}Z_r^2 = - g_r.
\end{equation*}
Consequently we have the usual renormalization conditions
\begin{equation*}
  \opifmr 2(0)=\opifr 2(0)=0,\quad \frac{\partial \opifr 2}{\partial
    n^2}(0)=0,\quad\opifmonor 4 (0,0)=0.
\end{equation*}

The full \emph{melonic} two-point function is therefore
\begin{equation*}
  \cof 2(\ntup) = \frac{\kappa_{\jm}(\ntup)}{\opifbare 2(\ntup)} =
  \frac{\kappa_{\jm}(\ntup)}{m_r^2  +\opifmr 2(\ntup)} = \frac{\kappa_{\jm}(\ntup)}{\ntup^2 + m_{r}^2+\opifr 2 (\ntup)} 
\end{equation*}
so that in particular $\cof 2(0) = m_r^{-2}$.

\subsection{Mass renormalization}

Let us start by performing the mass renormalization, and postpone the wave-function and four-point 
coupling constant renormalization to the next section. Indeed mass renormalization is simpler as it does not involve
renormalons \cite{Riv1}. So throughout this section we keep the bare coupling constant $g_b$, and the bare wave-function normalization $Z_b$.

The mass renormalization subtracts recursively the value of all
subinsertions at $0$ external momentum. Hence, recalling \cref{eq-defGamma2}, the monochrome
melonic mass-renormalized self-energy $\selfnrjmono[mr]$ obeys the closed equation
\begin{equation*}
  \selfnrjmono[mr](n_c) =  \selfnrjmono (n_c) -  \selfnrjmono (0) =- g_bZ_b^2
  \sum_{\ptup \in\Z^5} \kappa_{\jm}(\ptup)\frac{\delta_{p_c,n_c} - \delta_{p_c,0}}{Z_b \ptup^2 + m_{r}^2  - 
    \sum_{c'} \selfnrjmono[mr](p_{c'}) }.
\end{equation*}
The sum over $\ptup$ in the equation above diverges only
logarithmically as $\jm \to \infty$. The total mass counterterm is 
\begin{align*}
  \delta_m &=  m_r^2 - Z_b m_b^2  = g_b Z_b^2 \sum_c  \delta_m^c,   \\
  \delta_m^c &= \sum_{\ptup \in \Z^5}\frac{\kappa_{\jm}(\ptup)\, \delta_{p_c,0}}{Z_b \ptup^2 +
    m_{r}^2-\sum_{c' } \selfnrjmono[mr](p_{c'})}=\sum_{\ptup \in \Z^4} \frac{\kappa_{\jm}(0,\ptup)}{Z_b \ptup^2 + m_{r}^2-\sum_{c'\neq c}
    \selfnrjmono[mr](p_{c'})}
\end{align*}
where we used that $\selfnrjmono[mr](0)=0$. Remark that $ \delta_m^c $ is independent of $c$, so that in fact 
\begin{equation*}
  \delta_m = 5 g_b Z_b^2 \sum_{\ptup \in \Z^4} \frac{\kappa_{\jm}(0,\ptup)}{Z_b \ptup^2 + m_{r}^2-\sum_{c=1}^{4}
    \selfnrjmono[mr](p_{c})}.
\end{equation*}

\subsection{Effective renormalization}
\label{sec-wave-funct-renorm}

We want to perform only the effective (or ``useful'') part of the coupling constant and 
wave-function renormalizations, that is when the
inner loop slice is higher than the external one.

\subsubsection{Multiscale decomposition}
\label{sec-mult-decomp}

After full mass renormalization, all correlation functions, as formal
power series in $g_{b}$, are expressed as sums over Feynman graphs
with mass-renormalized amplitudes containing mass-renormalized
propagators, \ie
\begin{equation*}
  C_{mr}(\ntup)\defi\frac{\kappa_{\jm}(\ntup^{2})}{Z_{b}\ntup^{2}+m_{r}^{2}}.
\end{equation*}
Multiscale analysis amounts to decompose the
(mass-renormalized here) propagator into a sum of sclice propagators
$\slicedpropa{j}_{mr}(\ntup)$ where for each
$j\in\set{0,1,\dotsc,\jm}$, $\slicedpropa{j}_{mr}(\ntup)$ ensures that
$\ntup^{2}$ is of order $M^{2j}$. Similarly to \cref{UV-cutoff} we define
\begin{align*}
  \kappa_{j}(\ntup^{2})&\defi\indic_{[-\frac 52,\frac
                         52]}\star\chi_{\frac
                         32}(M^{-2j}\ntup^{2})=\kappa(M^{-2j}\ntup^{2}),\quad 0\les j\les\jm\\
  \intertext{and decompose $\kappa_{\jm}$ as follows}
  \kappa_{\jm}&=\sum_{j=0}^{\jm}\eta_{j},\quad
                \eta_{j}=\kappa_{j}-\kappa_{j-1}\text{ for $1\les
                j\les\jm$,}\quad\eta_{0}=\kappa_{0}.
\end{align*}
Provided $M^{2}>2$, $\eta_{j}$ ($j>0$) is positive, smooth, and satisfies
\begin{equation*}
  \eta_{j}(\ntup^{2})=
  \begin{cases}
    0&\text{if $\ntup^{2}<M^{-2}M^{2j}$ or $\ntup^{2}>2M^{2j}$,}\\
    1&\text{if $\frac 2{M^{2}}M^{2j}\les\ntup^{2}\les M^{2j}$.}
  \end{cases}
\end{equation*}
As a consequence, we define $\slicedpropa{j}_{mr}(\ntup)$ as
$\eta_{j}(\ntup^{2})\del{Z_{b}\ntup^{2}+m_{r}^{2}}^{-1}$ so that
$C_{mr}(\ntup)=\sum_{j=0}^{\jm}\slicedpropa{j}_{mr}(\ntup)$. The
decomposition of each propagator in the amplitude $A_{mr}(G)$ of any
Feynman graph $G$ allows to write
\begin{equation*}
  A_{mr}(G)=\sum_{\mu\in[\jm]^{e(G)}}A_{mr}^{\mu}(G).
\end{equation*}
$e(G)$ is the number of internal edges of $G$. $\mu$ is called a scale
attribution and corresponds to choosing one index $j_{\ell}$ in $[\jm]$ for each
internal edge $\ell$ of $G$ so that the corresponding propagator is $\slicedpropa{j_{\ell}}_{mr}$.

\subsubsection{Effective constants}
\label{sec-effective-constants}

It will be convenient to define some more cutoff functions: for $j\in[\jm]$
\begin{equation*}
  \eta_{\ges j}\defi\sum_{l=j}^{\jm}\eta_{l}=\kappa_{\jm}-\kappa_{j-1}.
\end{equation*}
\begin{defn}[Effective wave-function]
  \label{def-Zj}
  The effective wave-function constant $Z_j$ is
  \begin{equation*}
    Z_{j}\defi Z_{b}-\frac{\partial \selfnrjmonomrslice{\ges j+1}}{\partial
      n_{c}^{2}}(0)
  \end{equation*}
  where $\selfnrjmrslice{\ges j}(\ntup)=\sum_{c}\selfnrjmonomrslice{\ges
  j}(n_{c})$ is the sum of mass-renormalized amplitudes of all 1PI melonic
$2$-point graphs, all internal scales of which are greater than or equal
to $j$, namely
\begin{equation*}
    \selfnrjmonomrslice{\ges j}(n_{c})\defi - g_bZ_b^2 \sum_{\ptup
    \in\Z^5} \eta_{\ges j}(\ptup^{2})\frac{\delta_{p_c,n_c} - \delta_{p_c,0}}{Z_b \ptup^2 + m_{r}^2  - 
    \sum_{c'} \selfnrjmono[mr](p_{c'})}.
\end{equation*}
\end{defn}
Note that with these notations, $Z_{\jm}=Z_{b}$ and $Z_{-1}=Z_{r} = 1$.
\begin{defn}[Effective coupling constant]
  \label{def-effective-coupling}
  The effective coupling constant $g_{j}Z_{j}^{2}$ is
  \begin{align*}
    -g_{j}Z_{j}^{2}&\defi \opifmonobslice{4}{\ges j+1}(0,0)\\
    \intertext{where}
    \opifmonobslice{4}{\ges j}(n_{c},\bar n_{c})&\defi \frac{
                                                  -g_bZ_b^2}{1 +
                                                  g_bZ_b^2\sum_{\ptup,\qbtup}\delta_{\ptup_{\hat
                                                  c},\qbtup_{\hat
                                                  c}}\delta_{p_c,n_c}
                                                  \delta_{\bar q_c, \bar
                                                  n_c}
                                                  \cofmrslice{2}{\ges
                                                  j}(\ptup)
                                                  \cofmrslice{2}{\ges j}(\qbtup)}\\
  \intertext{and}
  \cofmrslice{2}{\ges j}(\ntup)&\defi \frac{\eta_{\ges j}(\ntup^{2})}{Z_b\ntup^2 +
                       m_r^2  -\selfnrjmrslice{\ges j} (\ntup)}.
\end{align*}
\end{defn}
With these conventions, $g_{\jm} =g_b$ and $g_{-1} =g_r$.

\subsection{Analyticity}
\label{sec-analyticity}

This section is devoted to proving that the effective wave-functions and coupling constants are analytic functions of the
bare coupling $g_{b}$ (in a disk of radius going to $0$ as $\jm\to\infty$).

According to \cref{fig-relationship-G2Gamma2}, the generating function
for the number of 1PI divergent $2$-point graphs is
\begin{equation*}
  \selfnrjGFb(x)=\sum_{n=1}^{\infty}5^{n-1}C_{n-1}\,x^{n},\qquad C_{n}=\frac{1}{n+1}\binom{2n}{n}.
\end{equation*}
This can be proven either by solving the associated equation for
$\selfnrjGFb$,
\begin{equation}\label{eq-GFeq-2pt-Gamma}
    \selfnrjGFb=\frac{x}{1-5\selfnrjGFb}\iff 5\del[1]{\selfnrjGFb}^2-\selfnrjGFb+x=0,
\end{equation}
or by noticing that divergent melonic $2$-point graphs of
order $n$ (the root-vertex of which has a fixed colour) are in bijection with rooted plane trees with $n-1$ edges with a choice of one colour among five per edge. As such trees are counted by the Catalan number $C_{n-1}$, we get the result.

According to \cref{eq-Gammabar4-expression-fct-G2}, the monochrome 1PI generating function $\opifGFb 4$ of divergent $4$-point graphs is given by
\begin{align*}
    \opifGFb 4 &= x\frac{\del[1]{1-5\selfnrjGFb}^2}{\del[1]{1-5\selfnrjGFb}^2-x}\\
    \intertext{which, from \cref{eq-GFeq-2pt-Gamma}, gives}
    \opifGFb 4 &=x\frac{1-5x-5\selfnrjGFb}{1-6x-5\selfnrjGFb}.
\end{align*}
By a very simple application of the transfer theorems of Flajolet and Sedgewick \cite[chapter VI]{Flajolet2009aa}, the coefficients of $\opifGFb 4$ are asymptotically equal to $\frac{20^n}{64\sqrt{\pi n^3}}$.
\begin{hproof}
  The asymptotic value of the coefficients of $\opifGFb 4$ is fixed by its behaviour near its dominant singularity, namely the one closest to the origin. From \cref{eq-GFeq-2pt-Gamma},
  \begin{equation*}
    \opifGFb 2(x)=\frac 1{10}\del[1]{1-\sqrt{1-20x}}
  \end{equation*}
  so that $\opifGFb 4$ has a branch point at $x=1/20$ and other possible singularities at the roots of $1-6x-\tfrac 12\del{1-\sqrt{1-20x}}=0\iff \sqrt{1-20x}=12x-1$. It turns out that this last equation has no complex solution on $\C\setminus\R_{\ges\frac 1{20}}$ (but it has two roots on the second sheet of the Riemann surface of the square root). It can indeed be checked that $1-20x=(12x-1)^2$ has two roots, $0$ and $\frac 1{36}$, which are solutions of $-\sqrt{1-20x}=12x-1$. As a consequence, the asymptotic behaviour of the coefficients of $\opifGFb 2$ is given by the coefficients of its expansion around $\frac 1{20}$. But
  \begin{equation*}
    \opifGFb 4=\frac 1{16}-\frac 1{32}\sqrt{1-20x}+\bO(1-20x)
  \end{equation*}
  so that
  \begin{equation*}
    [x^n]\opifGFb 4\sim_{n\to\infty} \frac{20^n}{64\sqrt{\pi n^3}}.
  \end{equation*}
\end{hproof}

Remembering the definitions of $Z_j$ and $g_jZ_j^2$ (see \cref{sec-effective-constants}), we have
\begin{align}
      Z_j&\defi Z_{b}-\frac{\partial}{\partial
           n_{c}^{2}}\selfnrjmonomrslice{\ges j+1}(0)=Z_b+\sum_{n=1}^\infty (g_bZ_{b}^{2})^n A_n(m_r^2,Z_{b},\jm,j),\label{eq-def-An}\\
  g_jZ_j^2&\defi -\opifmonobslice{4}{\ges j+1}(0,0)=\sum_{n=1}^{\infty}(g_bZ_b^2)^nB_n(m_r^2,Z_b,\jm,j).\label{eq-def-Bn}
\end{align}
$A_n$ is the sum of the derivatives of the mass-renormalized amplitudes of the 1PI
divergent melonic $2$-point graphs of order $n$. $B_n$ is the sum of
the mass-renormalized amplitudes of the 1PI divergent melonic
$4$-point graphs of order $n$. According to their generating
functions, the number of such graphs is bounded by a constant to the
power $n$. Moreover there certainly exist $p,q\in\N$ such that the amplitudes of these graphs are bounded by $(\jm)^{pn}M^{2qn\jm}$.

Recall that
\begin{equation*}
  Z_{b}=1+\frac{\partial\selfnrjmono}{\partial n_{c}^{2}}(0)=1+\frac{\partial\selfnrjmono[mr]}{\partial n_{c}^{2}}(0).
\end{equation*}
Then, by the implicit function theorem, $Z_{b}$ is an analytic
function of $g_{b}$ in a neighbourhood of $0$ (which shrinks to
$\set{0}$ as $\jm\to\infty$). Let us now define $F$ and $G$ on
$\Omega_{1}\times\Omega_{2}$ where $\Omega_{1}$ (\resp $\Omega_{2}$)
is a complex neigbourhood of $0$ (\resp of $1$) such that
\begin{equation*}
  Z_{j}=F(g_{b},Z_{b})\quad\text{and}\quad g_{j}Z_{j}^{2}=G(g_{b},Z_{b}).
\end{equation*}
The amplitude of any divergent graph is a finite sum (because our
UV cutoff is compactly supported) of analytic
functions of $Z_{b}$ in $\Omega_{2}$. $A_{n}$ and $B_{n}$ are thus
analytic functions of $Z_{b}$. Series in \cref{eq-def-An,eq-def-Bn}
converge normally so that $F$ and $G$ are analytic on
$\Omega_{1}\times\Omega_{2}$. Finally $Z_{j}$ and $g_{j}Z_{j}^{2}$ are
holomorphic functions of $g_{b}$ around $0$, by composition of $F$ and
$G$ respectively with $Z_{b}(g_{b})$. This proves that, at fixed UV cut-off $\jm$, both $g_jZ_j^2$ and $Z_j$ are analytic functions of $g_b$ in a neighbourhood of $0$.

Note also that $g_j$ is an invertible function of $g_b$ in a neighbourhood of $0$.

\subsection{Asymptotic freedom}
\label{sec-asymptotic-freedom}

Our aim is to prove
\begin{thm}
  \label{thm-asymptotic-freedom}
  For all $j\in\set{-1,0,\dotsc,\jm -1}$,
  \begin{equation*}
  g_{j+1}-g_{j}=\beta_{j} g_{j}^{2}+\bO(g_{j}^{3})
\end{equation*}
where $\beta_{j}=\beta_{2}+\bO(M^{-j})$, $\beta_{2}$ is a negative
real number and $\bO(g_j^3)=g_j^3f(g_j)$ where $f$ is analytic around
the origin (in a domain which shrinks to $\set{0}$ as $\jm\to\infty$).
\end{thm}
\begin{proof}
  Let us define $\sbe{\alpha}{$j$}_{1},\sbe{\alpha}{$j$}_{2}$ and
  $\sbe{\gamma}{$j$}_{1}$ as coefficients of the Taylor expansions of
  $g_{j}Z_{j}^{2}$ and $Z_{j}$:
  \begin{equation*}
    g_{j}Z_{j}^{2}\fide \sbe{\alpha}{$j$}_{1}g_{b}+\sbe{\alpha}{$j$}_{2}g_{b}^{2}+\bO(g_{b}^{3}),\qquad
    Z_{j}\fide 1+\sbe{\gamma}{$j$}_{1}g_{b}+\bO(g_{b}^{2}).
  \end{equation*}
  We thus have
  \begin{equation*}
    g_{j}=\sbe{\alpha}{$j$}_{1}g_{b}+(\sbe{\alpha}{$j$}_{2}-2\sbe{\alpha}{$j$}_{1}\sbe{\gamma}{$j$}_{1})g_{b}^{2}+\bO(g_{b}^{3})\iff
    g_{b}=\frac{1}{\sbe{\alpha}{$j$}_{1}}g_{j}-\frac 1{\big(\sbe{\alpha}{$j$}_{1}\big)^{3}}(\sbe{\alpha}{$j$}_{2}-2\sbe{\alpha}{$j$}_{1}\sbe{\gamma}{$j$}_{1})g_{j}^{2}+\bO(g_{j}^{3}).
  \end{equation*}
  Inserting the previous equation into the Taylor expansion of
  $g_{j+1}$ at order $2$, we get
  \begin{equation*}
    g_{j+1}=\frac{\sbe{\alpha}{$j+1$}_{1}}{\sbe{\alpha}{$j$}_{1}}g_{j}-\frac{1}{\big(\sbe{\alpha}{$j$}_{1}\big)^{2}}\Big[\frac{\sbe{\alpha}{$j+1$}_{1}}{\sbe{\alpha}{$j$}_{1}}(\sbe{\alpha}{$j$}_{2}-2\sbe{\alpha}{$j$}_{1}\sbe{\gamma}{$j$}_{1})-\sbe{\alpha}{$j+1$}_{2}+2\sbe{\alpha}{$j+1$}_{1}\sbe{\gamma}{$j+1$}_{1}\Big]g_{j}^{2}
    +\bO(g_{j}^{3}).
  \end{equation*}

  Let us now compute the coefficients
  $\sbe{\alpha}{$j$}_{1},\sbe{\alpha}{$j$}_{2}$ and
  $\sbe{\gamma}{$j$}_{1}$:
  \begin{multline*}
    -g_{j}Z_{j}^{2}= \opifmonobslice{4}{\ges
      j+1}(0,0)=-g_{b}Z_{b}^{2}+(g_{b}Z_{b}^{2})^{2}\sbe{\cA}{$j$}_{4,2}(0,0)+\bO(g_{b}^{3}),\\
    \sbe{\cA}{$j$}_{4,2}(n_{c},\bar n_{c})=\sum_{p\in\Z^{4}}C_{\ges
      j+1}(p,n_{c}) C_{\ges j+1}(p ,\bar n_{c})
  \end{multline*}
  where
  $C_{\ges j+1}(p)\defi \eta_{\ges j+1}(p^{2})/(Z_{b}p^{2}+m_{r}^{2})$
  and $\sbe{\cA}{$j$}_{4,2}(n_{c},\bar n_{c})$ is the
  mass-renormalized amplitude, ``down to scale $j$'', of the rightmost
  graph of \cref{fig-melonicdivergences}. To get
  $\sbe{\alpha}{$j$}_{1}$ and $\sbe{\alpha}{$j$}_{2}$, we need the
  Taylor expansion of $Z_{b}$ at first order:
  \begin{equation*}
    Z_{b}\fide 1+\sbe{\gamma}{$-1$}_{1}g_{b}+\bO(g_{b}^{2})\implies -g_{j}Z_{j}^{2}=-g_{b}+(\sbe{A}{$j$}_{4,2}(0,0)-2\sbe{\gamma}{$-1$}_{1})g_{b}^{2}+\bO(g_{b}^{3})
  \end{equation*}
  where $\sbe{A}{$j$}_{4,2}$ equals $\sbe{\cA}{$j$}_{4,2}$ evaluated
  at $Z_{b}=1$. We have thus
  \begin{equation*}
    \sbe{\alpha}{$j$}_{1}=1,\quad\sbe{\alpha}{$j$}_{2}=2\sbe{\gamma}{$-1$}_{1}-\sbe{A}{$j$}_{4,2}(0,0).
  \end{equation*}
  Before computing the flow equation for $g_{j}$, we need the first
  order Taylor coefficient $\sbe{\gamma}{$j$}_{1}$ of $Z_{j}$:
  \begin{align*}
    Z_{j}&=Z_{b}-\frac{\partial \selfnrjmonomrslice{\ges j+1}}{\partial
           n_{c}^{2}}(0)=Z_{b}+g_{b}Z_{b}^{2}\frac{\partial}{\partial n_{c}^{2}}\sum_{p\in\Z^{4}}\frac{\eta_{\ges
           j+1}(p^{2}+n_{c}^{2})}{Z_{b}(p^{2}+n_{c}^{2})+m_{r}^{2}}\Big\vert_{n_{c}^{2}=0}+\bO(g_{b}^{2})\\
         &=Z_{b}-g_{b}Z_{b}^{3}\sum_{p\in\Z^{4}}\frac{\eta_{\ges j+1}(p^{2})}{(Z_{b}p^{2}+m_{r}^{2})^{2}}+g_{b}Z_{b}^{2}\sum_{p\in\Z^{4}}\frac{\eta'_{\ges
           j+1}(p^{2})}{Z_{b}p^{2}+m_{r}^{2}}+\bO(g_{b}^{2})\\
         &\fide
           Z_{b}-g_{b}Z_{b}^{3}\sbe{\tilde{\cA}}{j}_{4,2}(0,0)+g_{b}Z_{b}^{2}\cK_{j}+\bO(g_{b}^{2})\\
         &=1+\del{\sbe{\gamma}{-1}_{1}-\sbe{\tilde{A}}{j}_{4,2}(0,0)+K_{j}}g_{b}+\bO(g_{b}^{2})\\
         &\Longrightarrow \sbe{\gamma}{j}_{1}=\sbe{\gamma}{-1}_{1}-\sbe{\tilde{A}}{j}_{4,2}(0,0)+K_{j}
  \end{align*}
  where, once again, $\sbe{\tilde{A}}{j}_{4,2}$ (\resp $K_{j}$) equals
  $\sbe{\tilde{\cA}}{j}_{4,2}$ (\resp $\cK_{j}$) evaluated at
  $Z_{b}=1$. Finally, we get
  \begin{multline*}
    g_{j+1}-g_{j}=-\big[-\del{\sbe{A}{$j$}_{4,2}(0,0)-\sbe{A}{$j+1$}_{4,2}(0,0)}+2\del{\sbe{\tilde{A}}{$j$}_{4,2}(0,0)-\sbe{\tilde{A}}{$j+1$}_{4,2}(0,0)}\\
    -2K_{j}+2K_{j+1}\big]g_{j}^{2}+\bO(g_{j}^{3}).
  \end{multline*}
  We now prove that $K_{j}$ (like $K_{j+1}$) is of order $M^{-2j}$ and
  that the sum of the other terms in $\beta_{j}$ equals a positive
  constant plus $\bO(M^{-j})$. First, we note that
  $\eta_{j}(p^{2})=h(M^{-2j}p^{2})$ where $h(p^{2})=\kappa(p^{2})-\kappa(M^{2}p^{2})$. Remark also that the
                       support of $h$ is $[M^{-2},2]$.
  \begin{align*}
    K_{j}&=\sum_{p\in\Z^{4}}\frac{\eta'_{j+1}(p^{2})}{p^{2}+m_{r}^{2}}=M^{-2(j+1)}\sum_{p\in\Z^{4}}\frac{h'(M^{-2(j+1)}p^{2})}{p^{2}+m_{r}^{2}}\\
    &=\bO(M^{-2j})+\sum_{p\in(M^{-j-1}\Z)^{4}}M^{-4(j+1)}\frac{h'(p^{2})}{p^{2}}.
  \end{align*}
  The above sum is a Riemann sum of the compactly supported $C^{1}$
  function $h(p^{2})/p^{2}$. Its difference with the corresponding
  integral (which vanishes) is of order of the mesh, that is
  $M^{-j}$. Thus $K_{j}=\bO(M^{-j})$.
  \begin{align*}
    \tilde{A}\defi\sbe{\tilde{A}}{$j$}_{4,2}(0,0)-\sbe{\tilde{A}}{$j+1$}_{4,2}(0,0)&=\sum_{p\in\Z^{4}}\frac{\eta_{j+1}(p^{2})}{(p^{2}+m_{r}^{2})^{2}}=\bO(M^{-2j})+M^{-4(j+1)}\sum_{p\in(M^{-j-1}\Z)^{4}}\frac{h(p^{2})}{p^{4}}\\
                                                                                   &=\bO(M^{-j})+\int_{\R^{4}}\frac{h(p^{2})}{p^{4}}\,d^{4}p,\\
    A\defi\sbe{A}{$j$}_{4,2}(0,0)-\sbe{A}{$j+1$}_{4,2}(0,0)&=\sum_{p\in\Z^{4}}\frac{\eta_{\ges
                                                       j+1}^{2}(p^{2})-\eta_{\ges
                                                             j+2}^{2}(p^{2})}{(p^{2}+m_{r}^{2})^{2}}\\
                                                                                   &=\sum_{p\in\Z^{4}}\frac{\eta_{
                                                                                     j+1}^{2}(p^{2})+2\eta_{j+1}(p^{2})\eta_{
                                                                                     j+2}(p^{2})}{(p^{2}+m_{r}^{2})^{2}}\\
                                                                                   &=\bO(M^{-j})+\int_{\R^{4}}\frac{h^{2}(p^{2})+2h(p^{2})h(M^{-2}p^{2})}{p^{4}}\,d^{4}p
  \end{align*}
  where we used $\eta_{\ges j+1}=\eta_{j+1}+\eta_{\ges j+2}$ and
  $\eta_{i}\eta_{j}=0$ if $|i-j|>1$. We get
  \begin{align*}
    \beta_{j}=\beta_{2}+\bO(M^{-j}),\quad\beta_{2}&\defi
                                                    -\int_{\R^{4}}\frac{d^{4}p}{p^{4}}\sbr{2h(p^{2})-h^{2}(p^{2})-2h(p^{2})h(M^{-2}p^{2})}\\
    &=-\int_{\R^{4}}\frac{d^{4}p}{p^{4}}h(p^{2})\sbr{2\del{1-\kappa(M^{-2}p^{2})}+\kappa(p^{2})+\kappa(M^{2}p^{2})}<0.
  \end{align*}
  The analyticity of $g_{j+1}-g_{j}-\beta_{j}g_{j}^{2}$ as a function
  of $g_{j}$ follows from the analyticity of $g_{j}$ and $Z_{j}$ as
  functions of $g_{b}$ (see \cref{sec-analyticity}).
\end{proof}
We have proven that for all $j$, $g_{j}$ is a holomorphic function of
$g_{b}$ in a neighbourhood of $0$ which goes to $\set 0$ as
$\jm\to\infty$. This defines $g_{j+1}$ as a holomorphic function of
$g_{j}$, in a neighbourhood of $0$ which goes to $\set 0$ as
$\jm\to\infty$. Moreover the first two coefficients of the expansion
$g_{j+1}$ in powers of $g_{j}$ have a finite limite as $\jm\to\infty$.
\begin{hquestion}
  On a montré que $\forall j, g_{j}$ est une fonction analytique de
  $g_{b}$ dans un voisinage de $0$ qui tend vers $\varnothing$ quand
  $\jm\to\infty$. Cela définit $g_{j+1}$ comme fonction analytique de
  $g_{j}$, dans un voisinage de $0$ qui tend vers $\varnothing$ quand
  $\jm\to\infty$. Or les deux premiers coefficients du développement
  de $g_{j+1}$ en puissances de $g_{j}$ ont une limite (finie) quand
  $\jm\to\infty$. Est-ce le cas à tout ordre? La série naturelle à
  considérer est $g_{j}\del{\set{g_{l},Z_{l},\,l\ges j+1}}$ mais quid de $g_{j+1}(g_{j})$?
\end{hquestion}

\section{Holomorphic RG flow}
\label{sec-holomorphic-rg-flow}

In \cref{sec-asymptotic-freedom} we proved that
\begin{equation*}
  g_{j+1}=g_{j}+\beta_{j}g_{j}^{2}+g_{j}^{3}f(g_{j})\fide h_{\jm,j}(g_{j})
\end{equation*}
where $f$ is holomorphic on a neighbourhood $\Omega_{\jm}$ of the origin and
$\beta_{j}=\beta_{2}+\bO(M^{-j})$, $\beta_{2}<0$. Note that \apriori
$\Omega_{\jm}\to\set{0}$ as $\jm\to\infty$. But the first two
Taylor coefficients of $h_{\jm,j}$ have in fact finite limits as the
ultraviolet cutoff is removed. In order to know if such a result holds
true at all orders, which would prove that $h_{\jm,j}$ is holomorphic in a
domain uniform in $\jm$, we need a better understanding of the series
$g_{j+1}(g_{j})$. In the sequel, we assume it.
\begin{assumption}
  The series $g_{j+1}(g_{j})$ is holomorphic in a domain uniform in $\jm$.
\end{assumption}
The dynamics defined by $h_{\jm,j}$ is not autonomous, its Taylor
coefficients depend on $j$. Nevertheless, far from the infrared cutoff
(here $m_{r}^{2}$), the behaviour of $\beta_{j}$ suggests that the
dynamics becomes autonomous. In the sequel, we assume it.
\begin{assumption}
  The discrete RG flow $g_{j+1}=h(g_{j})$ is defined by the iteration
  of a (unique) holomorphic map $h$, tangent to the identity, and such
  that
  \begin{equation}
    \label{eq-approx-RG-flow}
    h(z)=z+\beta_{2}z^{2}+\bO(z^{3}),\quad \beta_{2}<0.
  \end{equation}
\end{assumption}

Throughout this section, we will be interested in Cauchy problems with
complex initial data. In particular, we will prove appropriate uniform boundedness of their
solutions with respect to their initial data. In other words, we would
like to approach results such as ``for all $\epsilon>0$, there exists a complex
neighbourhood $\Omega_{\epsilon}$ of $0$ such that $g_{r}=g_{-1}\in\Omega_{\epsilon}$
implies for all $j\ges 0$, $|g_{j}|<\epsilon$''.

\subsection{Parabolic holomorphic local dynamics}
\label{sec-parab-holom-local-dynamics}

Our first objective is to understand the qualitative behaviour of the approximate RG
flow \eqref{eq-approx-RG-flow} by invoking the theory of holomorphic dynamical
systems. To this aim, we need to recall some classical definitions and
theorems, see \cite{Abate2010aa} for example.
\begin{defn}[Holomorphic dynamical system]
  \label{def-holo-dyna-sys}
  Let $M$ be a complex manifold, and $p\in M$. A (discrete)
  \emph{holomorphic local dynamical system} at $p$ is a holomorphic map
  $f:U\to M$ such that $f(p)=p$, where $U\subseteq M$ is an open
  neighbourhood of $p$; we shall assume that $f\neq\id_{U}$. We shall
  denote by $\EndMp$ the set of holomorphic local dynamical systems
  at $p$.
\end{defn}
We will only consider the case $M=\C$ and $p=0$.
\begin{defn}[Stable set]
  \label{def-stable-set}
  Let $f\in\EndMp$ be a \hlds defined on an open set $U\subseteq
  M$. Then the stable set $K_{f}$ of $f$ is
  \begin{equation*}
    K_{f}\defi\bigcap_{k=0}^{\infty}f^{\circ(-k)}(U).
  \end{equation*}
\end{defn}
In other words, the stable set of $f$ is the set of all points $z\in
U$ such that the orbit $\set[0]{ f^{\circ k}(z)\tqs k\in\N}$ is well-defined. If
$z\in U\setminus\set[0]{K_{f}}$ , we shall say that $z$ (or its orbit)
escapes from $U$. Clearly, $p\in K_{f}$ and so the stable set is never
empty (but it can happen that $K_{f}=\set{p}$).
\begin{defn}[Conjugation]
  We say that $f,g\in\EndCO$ are holomorphically conjugated if there
  exists a holomorphic map $h$ such that $h\circ f=g\circ h$.
\end{defn}
\begin{defn}[Classification]
  Let $f\in\EndCO$ be given by $f(z)=\lambda z+\sum_{j\ges
    2}a_{j}z^{j}$. We say that $f$ is
  \begin{enumerate}
  \item hyperbolic if $|\lambda|\neq 1$,
  \item parabolic if $\lambda^{q}=1$ for some
    $q\in\N\setminus\set{0}$,
  \item elliptic if $|\lambda|=1$ and $\lambda^{q}\neq 1$ for all $q\in\N\setminus\set{0}$.
  \end{enumerate}
\end{defn}
The RG flow we consider here is thus a parabolic dynamical system
($\lambda=1$).
\begin{defn}[Multiplicity]
  Let $f\in\EndCO$ be a \hlds with a parabolic fixed point at the
origin. Then we can write:
\begin{equation*}
  f(z)=e^{2i\pi p/q}z+a_{r+1}z^{r+1}+\bO(z^{r+2}),
\end{equation*}
with $a_{r+1}\neq 0$. $r+1$ is called the multiplicity of $f$.
\end{defn}
\begin{defn}[Directions]
  \label{def-directions}
  Let $f\in\EndCO$ be tangent to the identity of multiplicity $r+1\ges
  2$. Then a unit vector $v\in\bbS^{1}$ is an attracting (\resp
  repelling) direction for $f$ at the origin if $a_{r+1}v^{r}$ is real
  negative (\resp real positive).
\end{defn}
Clearly, there are $r$ equally spaced attracting directions, separated
by $r$ equally spaced repelling directions: if $a_{r+1}=|a_{r+1}|e^{i
  \alpha}$, then $v=e^{i\theta}$ is attracting (\resp repelling) \ifft
\begin{equation*}
  \theta=\frac{2k+1}{r}\pi-\frac{\alpha}r\quad \del{\text{\resp $\theta=\frac{2k}{r}\pi-\frac{\alpha}r$}}.
\end{equation*}
It turns out that to every attracting direction is associated a
connected component of $K_{f}\setminus\set{0}$.
\begin{defn}[Basins]
  \label{def-basin}
  Let $v\in\bbS^{1}$ be an attracting direction for an $f\in\EndCO$
  tangent to the identity. The basin centerd at $v$ is the set of
  points $z\in K_{f}\setminus\set{0}$ such that
  $\lim_{k\to\infty}f^{\circ k}(z)=0$ and
  $\lim_{k\to\infty}f^{\circ k}(z)/|f^{\circ k}(z)|=v$. If $z$ belongs to the
  basin centered at $v$, we shall say that the orbit of $z$ tends to
  $0$ tangent to $v$.
\end{defn}
\begin{defn}[Petals]
  \label{def-petal}
  An attracting petal centered at an attracting direction $v$ of an
  $f\in\EndCO$ tangent to the identity is an open simply connected
  $f$-invariant set $P\subseteq K_{f}\setminus\set{0}$ such that a point
  $z\in K_{f}\setminus\set{0}$ belongs to the basin centered at $v$ \ifft
  its orbit intersects $P$. In other words, the orbit of a point tends
  to $0$ tangent to $v$ if and only if it is eventually contained in
  $P$. A repelling petal (centered at a repelling direction) is an attracting petal for the inverse of $f$.
\end{defn}
\begin{thm}[Leau-Fatou flower]
  \label{thm-leau-fatou}
  Let $f\in\EndCO$ be a \hldsttti with multiplicity $r+1\ges 2$ at the
  fixed point. Let $v_{1}^{\pm},\dotsc,v_{r}^{\pm}\in\bbS^{1}$ be the
  attracting (\resp repelling) directions of $f$ at the origin. Then,
  \begin{enumerate}
  \item for each attracting (\resp repelling) direction $v_{j}^{\pm}$
    there exists an attracting (\resp repelling petal) $P_{j}^{\pm}$,
    so that the union of these $2r$ petals together with the origin
    forms a neighbourhood of the origin. Furthemore, the $2r$ petals
    are arranged cyclically so that two petals intersects \ifft the
    angle between their central directions is $\pi/r$.
  \item $K_{f}\setminus\set 0$ is the (disjoint) union of the basins
    centered at the $r$ attracting directions.
  \item If $B$ is a basin centered at one of the attracting
    directions, then there is a function $\varphi:B\to\C$ such that
    $\varphi\circ f(z)=\varphi(z)+1$ for all $z\in B$. Furthermore if $P$ is
    the corresponding petal, then $\varphi|_{P}$ is a biholomorphism with
    an open subset of the complex plane containing a right-half plane
    -- and so $f|_{P}$ is holomorphically conjugated to the
    translation $z\mapsto z+1$.
  \end{enumerate}
\end{thm}
As a consequence of \cref{thm-leau-fatou}, if $z$ belongs to an
attracting petal $P$ of a holomorphic local dynamical system tangent to
the identity, then its entire orbit is contained in $P$ and moreover
$f^{\circ n}(z)$ goes to $0$ (as $n\to\infty$), tangentially to the
corresponding attracting direction. A typical trajectory can be seen
on \cref{fig-flower-mult-4}. 
\begin{figure}[!htp]
  \centering
  \includegraphics[scale=1.3]{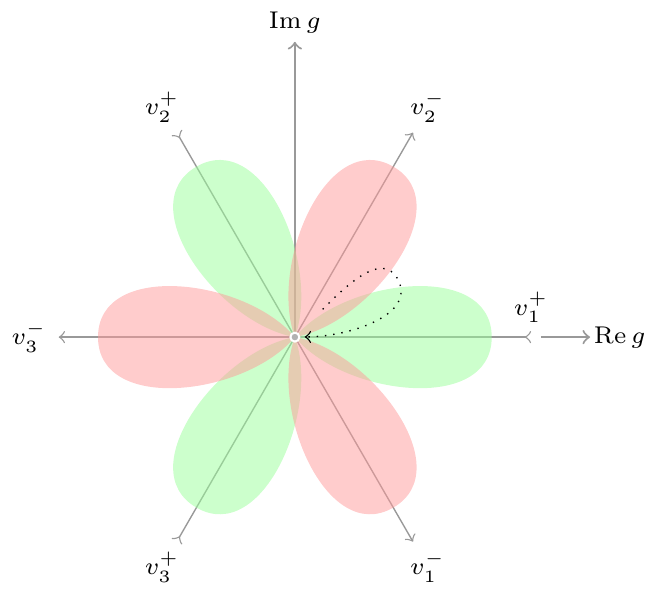}
  \caption{Attracting (green) and repelling (red) petals of a dynamics of
    multiplicity $4$, and a typical trajectory.}
  \label{fig-flower-mult-4}
\end{figure}
Note that \cref{thm-leau-fatou} asserts the existence of attracting
and repelling petals whose union with the origin forms a neighbourhood
of the origin. The intersection properties of these petals implies
that their asymptotic opening angle (\ie their opening angle close to
$0$) is strictly bigger than $\pi/r$ for a system of multiplicity
$r$. But in fact, with a bit more work, one can construct petals whose
asymptotic opening angle is $2\pi/r$, see \cite{Carleson1993aa}. Such
attracting petals are tangent at $0$ to their two neighbouring
repelling directions.\\

In case of the system \eqref{eq-approx-RG-flow}, we have a parabolic
dynamical system of multiplicity $2$ so that there is only one
attracting (\resp repelling) petal corresponding to the attractive
(\resp repelling) direction $(1,0)$ (\resp $(-1,0)$). The asymptotic
opening angle of the attracting petal is $2\pi$ which makes it very similar to cardioid-like
domains obtained by Loop Vertex Expansion
\cite{Delepouve2014aa,Rivasseau2016aa}, see \cref{fig-petal-2}.
\begin{figure}[!htp]
    \centering
    \begin{tikzpicture}[scale=2.5,line cap=round]
      \draw[->, black!40, semithick] (0:1.45cm) -- (0:1.65cm);
      \node at (0:1.8cm) {$\scriptstyle \Re g$};
      \draw[->, black!40, semithick] (90:0) -- (90:1.2cm);
      \node at (90:1.28cm) {$\scriptstyle \Im g$};
      \draw[-<, black!40] (0,0) -- (0:1.4cm);
      \node at (1.4cm,.1cm) {$\scriptstyle v^{+}$};
      \draw[->, black!40] (0,0) -- (0:-1.4cm);
      \node at (180:1.6cm) {$\scriptstyle v^{-}$};
      \fill[green!40,opacity=.5,xscale=1.2,yscale=.8] (0cm,1cm) arc
      [start angle=90, end angle=270, radius=.5cm] arc [start
      angle=90, end angle=270, radius=.5cm] arc [start angle=-90, end
      angle=90, radius=1cm];
      \node[vertexbw,scale=.7] at (0,0) {};
    \end{tikzpicture}
    \caption{A unique attracting petal of a multiplicity $2$ parabolic
      dynamics.}
    \label{fig-petal-2}
  \end{figure}
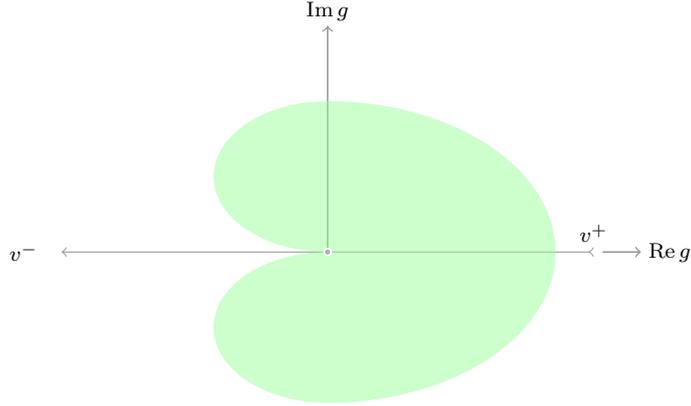

In the next sections, we get more quantitative results on the RG
trajectories in case $g_{r}$ is real, on the shapes of attracting
petals, and on the size of the Nevanlinna-Sokal disk they contain. To
this aim, we will study continuous dynamical systems, more precisely linear
ODEs, rather than iterations of holomorphic maps. This is justified by
the following argument. As we saw in \cref{sec-asymptotic-freedom},
there is evidence that the discrete RG flow of the $T^{4}_{5}$ tensor
field is well approximated by the dynamical system
\eqref{eq-approx-RG-flow}, at least in the deep ultraviolet. From
\cref{thm-leau-fatou} a trajectory starting in the unique attracting
petal $P_{+}$ remains forever in this petal. Moreover, in $P_{+}$, the
dynamics is conjugated to the translation $z\mapsto z+1$. But this
translation is the time $1$ flow of the constant vector field
$Y=\partial_{z}$. Thus there exists a holomorphic vector field $X$
such that $h$ in \cref{eq-approx-RG-flow} is the time $1$ flow of
$X$. The Taylor coefficients of $X$ can be computed recursively via
the equation $e^{X}(z)=h(z)$. One finds
$X=\del{\beta_{2}z^{2}+(\beta_{3}-2\beta_{2}^{2})z^{3}+\bO(z^{4})}\partial_{z}$
if $h(z)=z+\beta_{2}z^{2}+\beta_{3}z^{3}+\bO(z^{4})$. As a consequence
we will consider, in the next sections, ODEs of the form
\begin{equation*}
  g'=\beta_{2}g^{2}+\beta_{3}g^{3}+\bO(g^{4}), \quad\beta_{2}\in\R_{-}
\end{equation*}
keeping in mind that the above $\beta_{3}$ corresponds in fact to
$\beta_{3}-2\beta_{2}^{2}$ in the notation of \cref{eq-approx-RG-flow}.

\subsection{Quadratic flow}
\label{sec-quadratic-flow}

Let us first consider $f:\R\times\C\to\C$ such that
$f(t,z)=\beta_{2}z^{2}$, $\beta_{2}$ real negative, and the
following Cauchy problem:
\begin{subequations}
  \label{eq-continuous-quadratic-cpx-Cauchy-pb}
    \begin{align}
    g'&=f(t,g)\label{eq-continuous-quadratic-cpx-ODE}\\
    g(0)&=g_{r}\in\C.\label{eq-continuous-quadratic-cpx-initial-cond}
  \end{align}
\end{subequations}
The solution is obviously given by
\begin{equation*}
  \frac{1}{g(t)}=\frac{1}{g_{r}}-\beta_{2}t.
\end{equation*}
In polar
coordinates ($g=\rho e^{i\theta}$, $\theta\in[-\pi,\pi]$), we have
\begin{equation}\label{eq-solution-quadratic-cpx-flow-polar}
  \frac 1{g(t)}=\frac{e^{-i\theta}}{\rho}=\frac{\cos\theta_{r}}{\rho_{r}}-\beta_{2}t-i\frac{\sin\theta_{r}}{\rho_{r}}.
\end{equation}
If $\theta_{r}\in(-\pi,\pi)\setminus\set{0}$, $g$ is well-defined on
$\R$. If $\theta_{r}=0$, $g$ is well-defined on $\R_{+}$ (but explodes
at a finite negative time). If $\theta_{r}=\pi$, $g$ explodes at a finite (positive) time.\\

We now prove a uniform bound on $\module{g}$ for $g_{r}$ in the following compact
domain $\Omega_{\epsilon}$ of the complex plane.
\begin{defn}[Domain of uniform boundedness of a quadratic flow]
  \label{def-Omega}
  Let $\epsilon$ be a positive real number. In polar coordinates ($z=\rho
  e^{i\theta}$, $\rho\in\R_{+}$, $\theta\in[-\pi,\pi]$),
  \begin{equation*}
  \Omega_{\epsilon}\defi
  \begin{cases}
    \set{z\in\C\tqs \rho\les \epsilon}&\text{if
      $\theta\in[-\frac\pi 2,\frac\pi 2]$,}\\
    \set{z\in\C\tqs \rho\les \epsilon|\sin\theta|}&\text{if
      $|\theta|\in(\frac\pi 2,\pi]$.}
  \end{cases}
\end{equation*}
\end{defn}
The set $\Omega_{\epsilon}$ is made of three parts. On $\set{\Re z\ges 0}$, $\Omega_{\epsilon}$ is
a closed half-disk of radius $\epsilon$, centered at $0$. On $\set{\Re z\les
  0}\cap\set{\Im z\ges 0}$, $\Omega_{\epsilon}$ is a closed half-disk of radius $\frac
\epsilon 2$ centered at $i\frac \epsilon 2$. On $\set{\Re z\les
  0}\cap\set{\Im z\les 0}$, $\Omega_{\epsilon}$ is a closed half-disk of radius $\frac
\epsilon 2$ centered at $-i\frac \epsilon 2$. See \cref{fig-Omega} for a picture of $\Omega_{\epsilon}$.
\begin{figure}[!htp]
    \centering
    \begin{tikzpicture}[scale=2.5,line cap=round,domain=-180:180,samples=100]
      \draw[help lines,step=0.5cm] (-1.4,-1.4) grid (1.4,1.4);
      \draw[->] (-1.5,0) -- (1.5,0) node[right] {$\Re z$} coordinate(x
      axis);
      \draw[->] (0,-1.5) -- (0,1.5) node[above] {$\Im z$}
      coordinate(y axis); \draw[xshift=1 cm] (0pt,-1pt) -- (0pt,1pt)
      node[above right=.5pt,fill=white] {$\epsilon$};
      \draw[yshift=1 cm] (-1pt,0pt) -- (1pt,0pt) node[above
      right,fill=white] {$\epsilon$};
      \draw[yshift=-1 cm] (-1pt,0pt) -- (1pt,0pt) node[below
      right,fill=white] {$-\epsilon$};
      \foreach\y/\ytext in {-.5/-\frac{\epsilon}{2}, .5/\frac{\epsilon}{2}} \draw[yshift=\y
      cm] (1pt,0pt) -- (-1pt,0pt) node[left,fill=white] {$\ytext$};
      \filldraw[thick,fill=gray!50!white,opacity=.5] (0cm,1cm) arc
      [start angle=90, end angle=270, radius=.5cm] arc [start
      angle=90, end angle=270, radius=.5cm] arc [start angle=-90, end
      angle=90, radius=1cm];
      \draw[thin,red
      ] plot ({cos(.5*\x)*cos(.5*\x)*cos(\x)}, {cos(.5*\x)*cos(.5*\x)*sin(\x)});
    \end{tikzpicture}
    \caption{In gray, the domain $\Omega_{\epsilon}$ of \cref{def-Omega}. In red,
    the cardioid $\rho=\epsilon\cos^{2}(\tfrac{\theta}2)$.}
    \label{fig-Omega}
  \end{figure}
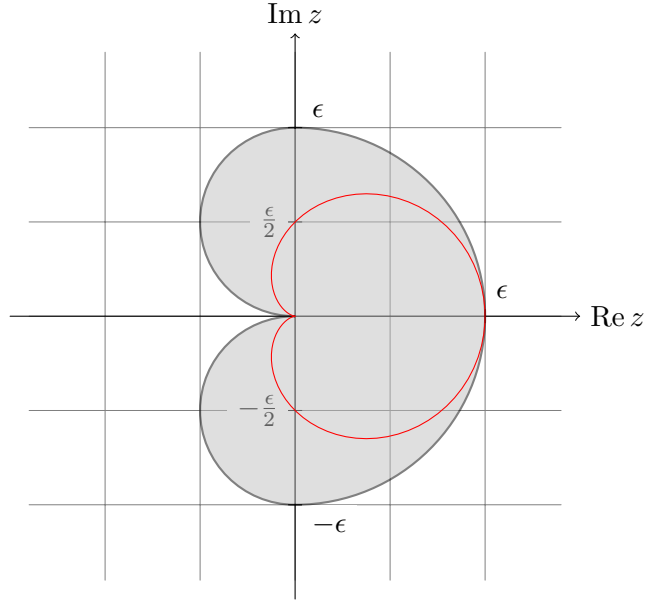
  \begin{rem}
    $\Omega_{\epsilon}$ contains the cardioid domain $\scC_{\epsilon}\defi\set{\rho\les
    \epsilon\cos^{2}(\tfrac{\theta}2)}$ which is the typical domain of
  analyticity of correlation functions predicted by Loop Vertex Expansion.
  \end{rem}

\begin{thm}
  \label{thm-continuous-quadratic-cpx-flow-boundedness}
  If $g_{r}\in\Omega_{\epsilon}$, then for all $t\in\R_{+}$,
  $\module{g(t)}\les \epsilon$.
\end{thm}
\begin{proof}
  From \cref{eq-solution-quadratic-cpx-flow-polar},
  \begin{equation}\label{eq-OneOverRho-complex-quadratic-flow}
  \frac{1}{\rho^{2}(t)}=\Big(\frac{\cos\theta_{r}}{\rho_{r}}-\beta_{2}t\Big)^{\!2}+\frac{\sin^{2}\theta_{r}}{\rho_{r}^{2}}.
\end{equation}
If $\theta_{r}\in[-\frac\pi 2,\frac\pi 2]$, $\cos\theta_{r}\ges 0$ and
$\rho$ attains its maximum at $t=0$ (recall that $\beta_{2}<0$) so
that $\rho(t)\les\rho_{r}$. If
$\module{\theta_{r}}\in(\frac\pi 2,\pi)$, $\cos\theta_{r}<0$ and
$\rho(t)\les\frac{\rho_{r}}{\module{\sin\theta_{r}}}$. This proves the desired bound.
\end{proof}
\begin{rem}
  In fact, by the holomorphic (on $\C^{*}$) change of coordinate
  $z\mapsto 1/z$, one can even prove that $g_{r}\in\Omega_{\epsilon}$ implies
  $g(t)\in\Omega_{\epsilon}$ for all $t>0$.
\end{rem}

\subsection{Cubic flow}
\label{sec-cubic-flow}

We now consider the following complex cubic differential flow:
\begin{subequations}\label{eq-complex-cubic-ODE-CauchyPb}
  \begin{align}
    g'&=\beta_{2}g^{2}+\beta_{3}g^{3}=\beta_{2}(x^{2}-y^{2})+\beta_{3}(x^{3}-3xy^{2})+2i\beta_{2}xy+i\beta_{3}(3x^{2}y-y^{3}),\label{eq-continuous-complex-cubic-ODE}\\
    g(0)&=g_{r}\in\C,
  \end{align}
\end{subequations}
where $x=\Re(g)$ and $y=\Im(g)$, $\beta_{2},\beta_{3}$ real with $\beta_{2}<0$. We have
\begin{thm}
  \label{thm-complex-cubic-ODE}
  For $\epsilon$ small enough, if $g_{r}\in\Omega_{\epsilon}$ then there exists a
  function
  \begin{align*}
    \phi:\R_{+}\times\Omega_{\epsilon}&\to\C\\
    (t,g_{r})&\mapsto\phi(t,g_{r})
  \end{align*}
  holomorphic in $g_{r}$, uniformly bounded, namely
  \begin{equation*}
    |\phi(t,g_{r})|<2\pi,\quad\text{for all $t\in\R_{+}$ and $g_{r}\in\Omega_{\epsilon}$},
  \end{equation*}
  such that the unique maximal solution of the Cauchy problem
  \eqref{eq-complex-cubic-ODE-CauchyPb} defined on $\R_{+}$ is
  \begin{equation*}
    g(t)=\frac{g_{r}}{1-\beta_{2}g_{r}t+\frac{\beta_{3}}{\beta_{2}}g_{r}\log(1-\beta_{2}g_{r}t)+\frac{\beta_{3}}{\beta_{2}}g_{r}\phi(t)}.
  \end{equation*}
\end{thm}
\begin{proof}
  The partial derivatives of the right-hand side of
  \cref{eq-continuous-complex-cubic-ODE} with respect to $x$ and $y$
  are continuous so that the Cauchy-Lipschitz theorem applies. As a
  consequence, for any complex initial data $g_{r}$, there exists a
  unique maximal continuously differentiable solution $g$ defined on
  $\intco{0,T}$ for some $T>0$. Moreover if $\Im g_{r}$ is positive
  (\resp negative), then for all $t\in\intco{0,T}$, $\Im g(t)$ is positive
  (\resp negative). In particular $g(t)\neq 0$.\\

    Let $g_{2}$ and $g_{3}$ be the two following complex functions on
  $\R_{+}$:
  \begin{align*}
    g_{2}(t)&=g_{r}\big(1-\beta_{2}g_{r}t\big)^{-1},\\
    g_{3}(t)&=g_{r}\Big(1-\beta_{2}g_{r}t+\frac{\beta_{3}}{\beta_{2}}g_{r}\log(1-\beta_{2}g_{r}t)\Big)^{-1}.
  \end{align*}
  $g_{2}$ is a solution of $g'=\beta_{2}g^{2}$ and $g_{3}$ a solution
  of $g'=\beta_{2}g^{2}+\beta_{3}g_{2}g^{2}$. Let us define the
  following new variables:
  \begin{equation*}
    u\defi\frac{g_{r}}{g},\quad u_{3}\defi\frac{g_{r}}{g_{3}},\quad
    \alpha\defi |\beta_{2}|g_{r},\quad \beta\defi\frac{\beta_{3}}{\beta_{2}}g_{r}.
  \end{equation*}
  In these new variables, \cref{eq-continuous-complex-cubic-ODE} rewrites as
  \begin{equation}
    \label{eq-continuous-complex-cubic-ODE-new}
    u'=\alpha \del{1+\frac{\beta}{u}}.
  \end{equation}
  Let us insert the ansatz
  $g^{-1}=g_{3}^{-1}+\frac{\beta_{3}}{\beta_{2}}\phi$ (or $
  u=u_{3}+\beta\phi$) into \cref{eq-continuous-complex-cubic-ODE-new} to get
  the ODE satisfied by $\phi$:
  \begin{equation}
    \label{eq-continuous-complex-cubic-ODE-phi}
    \phi'=-\frac{\alpha}{1+\alpha t}\frac{\beta\log(1+\alpha
      t)+\beta\phi}{u_{3}+\beta\phi}.
  \end{equation}
  The function $g$ of \cref{thm-complex-cubic-ODE} is a solution of the Cauchy problem
 \eqref{eq-complex-cubic-ODE-CauchyPb} \ifft $\phi$ satisfies
  \eqref{eq-continuous-complex-cubic-ODE-phi} with initial condition $\phi(0)=0$.\\
  
  Let $D$ be the following open subset of $\R\times\C$:
  \begin{equation*}
    \label{eq-def-D}
    D\defi\set{(t,z)\in (-\tfrac 1{\alpha},+\infty)\times\C\tqs u_{3}(t)+\beta z\neq 0}.
  \end{equation*}
  Let $h$ be the function defined by:
  \begin{equation*}
    \begin{aligned}
      h: D&\to\C\\
      (t,z)&\mapsto -\frac{\alpha}{1+\alpha t}\frac{\beta\log(1+\alpha
        t)+\beta z}{u_{3}+\beta z}.
    \end{aligned}  \label{eq-def-h}
  \end{equation*}
  It is easy to check that the partial derivatives of $h$ are
  continuous on $D$ so that $h$ is continuously differentiable on
  $D$. Then, as $(0,0)\in D$, by the Cauchy-Lipschitz theorem, there
  exists a unique (continuously differentiable) solution $\phi$ to $\phi'=h(t,\phi)$ such that
  $\phi(0)=0$. In particular, $\phi$ is defined on $[0,T)$ for some
  $T>0$.\\

Let us now show that $\phi$ is bounded (and thus defined in fact on $\R_{+}$). We integrate \cref{eq-continuous-complex-cubic-ODE-new} out: for $t\in\intco{0,T}$,
  \begin{equation*}
    \int_{0}^{t}\frac{uu'}{u+\beta}\diff
    t=\int_{u}\frac{z}{z+\beta}\diff
    z=u-1-\beta\log\del{\frac{u+\beta}{1+\beta}}=\alpha t
  \end{equation*}
  where we used $u(0)=1$. This implicit equation for $u$ can be turned
  into an implicit equation for $\phi$ under the ansatz $u=u_{3}+\beta\phi$:
  \begin{equation}
    \label{eq-implicit-phi-complex-cubic}
    \phi=\log\del{\frac{u_3+\beta\phi+\beta}{(1+\beta)(1+\alpha t)}}.
  \end{equation}
If there exists $\tau\in [0,T)$
  such that $|\phi(\tau)|>2\pi$, we define $t_{0}$ as the smallest
  $t\in [0,T)$ such that $|\phi(t)|=2\pi$. Thus on $[0,t_{0}]$, the
  continuous function $|\phi|$ is smaller than $2\pi$ and must take all values between $0$ and
 $2\pi$. But,
  \begin{align*}
    \phi&=-\log(1+\beta)+\log\del{1+\beta\frac{\log(1+\alpha
          t)+\phi+1}{1+\alpha t}},\\
    |\phi|&\les |\log(1+\beta)|+\envert[3]{\log\del{\envert{1+\beta\frac{\log(1+\alpha
            t)+\phi+1}{1+\alpha t}}}}+\pi\\
    &\les |\log(1+\beta)|+\envert[3]{\log\del{1+|\beta|\frac{|\log(1+\alpha
            t)|+|\phi+1|}{|1+\alpha t|}}}+\pi.
  \end{align*}
  If $|g_{r}|$ is small enough, $|1+\beta|$ is close to $1$ and
  $\arg(1+\beta)$ is small so that $|\log(1+\beta)|$ can be made
  smaller than $\tfrac 12$ (say). Let us recall that in polar
  coordinates, $g_{r}=\rho_{r}e^{i\theta_{r}}$. A simple inspection of $|1+\alpha
  t|$ as a function of $t\in\R_{+}$ shows that
  \begin{equation}\label{eq-min-f-of-t}
    |1+\alpha t|\ges
    \begin{cases}
      1&\text{if $|\theta_{r}|\in\intcc[0]{0,\tfrac\pi 2}$,}\\
      |\sin\theta_{r}|&\text{if $|\theta_{r}|\in\intoo[0]{\tfrac\pi 2,\pi}$.}
    \end{cases}
  \end{equation}
  As a consequence, for $g_{r}\in\Omega_{\epsilon}$ and $\epsilon$ small enough,
  \begin{equation*}
    |\beta|\frac{|\phi+1|}{|1+\alpha
      t|}\les\envert[2]{\frac{\beta_{3}}{\beta_{2}}}\frac{|g_{r}|}{|\sin\theta_{r}|}(|\phi|+1)\les
    \envert[2]{\frac{\beta_{3}}{\beta_{2}}}(2\pi+1) \epsilon\les\frac 13.
  \end{equation*}
  There remains to bound
  \begin{equation*}
    |\beta|\frac{|\log(1+\alpha t)|}{|1+\alpha t|}\les|\beta|\frac{\frac 12|\log|1+\alpha t|^{2}|+\pi}{|1+\alpha t|}.
  \end{equation*}
  Firstly,
  \begin{equation*}
    |\beta|\frac{\pi}{|1+\alpha t|}\les
    \pi\envert[2]{\frac{\beta_{3}}{\beta_{2}}}\frac{|g_{r}|}{|\sin\theta_{r}|}\les\frac 13.
  \end{equation*}
  Secondly, let us define the functions $f,g:\intco{0,T}\to\R$ by
  $f(t)=|1+\alpha t|^{2}$ and $g(t)=\frac{\log f(t)}{\sqrt{f(t)}}$. From
  \begin{align*}
    f'(t)=2\rho_{r}|\beta_{2}|\cos\theta_{r}+2t\rho_{r}^{2}|\beta_{2}|^{2}\quad\text{and}\quad
    g'(t)=\frac{f'(t)}{2f(t)^{3/2}}\del{2-\log f(t)},
  \end{align*}
  we deduce that $0\les g(t)\les \frac 2e$ so that
  \begin{equation*}
    \frac{|\beta|}{2}\frac{|\log|1+\alpha t|^{2}|}{|1+\alpha
      t|}\les\envert[2]{\frac{\beta_{3}}{\beta_{2}}}\frac{|g_{r}|}{e}\les\frac 13.
  \end{equation*}
  Thirdly,
  \begin{equation*}
    |\phi|\les\frac 12+\pi+\log 2 <2\pi
  \end{equation*}
  which is a contradiction and proves that
  $\enVert{\phi}_{\infty}<2\pi$. Finally, to prove that $\phi$ is a
  holomorphic function of $g_{r}$, we note that it is a solution of
  the implicit equation
  \begin{equation*}
  F_{t}(g_{r},z)=\frac{u_{3}+\beta z+\beta}{1+\beta}-(1+\alpha
  t)e^{z}=0
\end{equation*}
(see \cref{eq-implicit-phi-complex-cubic}). But $F_{t}$ is holomorphic
on $\del{\C\setminus\set[0]{-\frac{\beta_{2}}{\beta_{3}}}}\times\C$ so that
for $g_{r}\in\Omega_{\epsilon}$ and $\epsilon$ small enough, $\phi$ is holomorphic
in $g_{r}$ by the implicit function theorem.
\end{proof}

Let $\epsilon$ be a positive real number. Denoting the ratio
$\beta_{3}/\beta_{2}$ by $\beta_{3,2}$, we define $\cH_{\epsilon}$ as the following compact subset of $\C$:
\begin{equation*}
  \cH_{\epsilon}\defi
  \begin{cases}
      \set{z\in\C\tqs |z|\les\frac{\epsilon}{1+3\pi|\beta_{3,2}|\epsilon}}&\text{if $|\arg z|\in\intcc[0]{0,\frac\pi 2}$,}\\
      \set[1]{z\in\C\tqs |z|\les\frac{\epsilon|\sin\arg
        z|}{1+\epsilon|\beta_{3,2}|\del{\envert{\log|\sin\arg z|}+3\pi}}}&\text{if $|\arg
      z|\in\intoo[0]{\frac\pi 2,\pi}$.} 
\end{cases}
\end{equation*}
\begin{figure}[!htp]
    \centering
    \begin{tikzpicture}[scale=2.5,line cap=round,domain=90:270,samples=100]
      \draw[help lines,step=0.5cm] (-1.4,-1.4) grid (1.4,1.4);
      \draw[->] (-1.5,0) -- (1.5,0) node[right] {$\Re z$} coordinate(x
      axis);
      \draw[->] (0,-1.5) -- (0,1.5) node[above] {$\Im z$}
      coordinate(y axis);
      \draw[xshift=1 cm] (0pt,-1pt) -- (0pt,1pt)
      node[above right=.5pt,fill=white] {$\epsilon$};
      \draw[xshift=.67 cm] (0pt,-1pt) -- (0pt,1pt)
      node[above right,fill=white] {$\tilde \epsilon$};
      \filldraw[thick, fill=gray!50!white,opacity=.5] (0cm,-0.67cm) arc [start angle=-90, end
      angle=90, radius=0.67cm] plot ({cos(\x)*abs(sin(\x))/(1.5-0.054*ln(abs(sin(\x))))}, {sin(\x)*abs(sin(\x))/(1.5-0.054*ln(abs(sin(\x))))});
      \draw[thick,dashed,gray] (0cm,1cm) arc
      [start angle=90, end angle=270, radius=.5cm] arc [start
      angle=90, end angle=270, radius=.5cm] arc [start angle=-90, end
      angle=90, radius=1cm];
      \draw[thin,red, domain=-180:180,samples=100] plot ({.67cm*cos(.5*\x)*cos(.5*\x)*cos(\x)}, {.67cm*cos(.5*\x)*cos(.5*\x)*sin(\x)});
    \end{tikzpicture}
    \caption{In gray, the domain $\cH_{\epsilon}$ ($\tilde
      \epsilon\defi\frac{\epsilon}{1+3\pi|\beta_{3,2}|\epsilon}$, here
      $3\pi|\beta_{3,2}|=\frac 12$). The red (\resp dashed) line is the boundary of the
      cardioid $\scC_{\tilde\epsilon}$ (\resp of $\Omega_{\epsilon}$).}
    \label{fig-Heps}
  \end{figure}
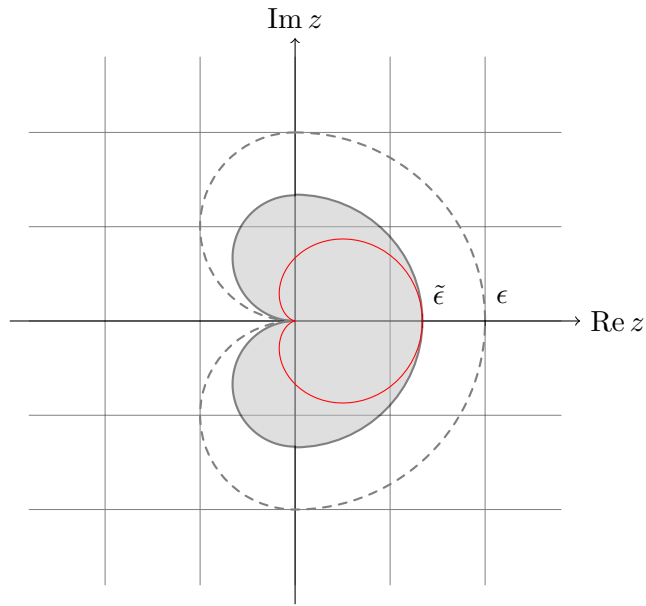
\begin{cor}
  \label{thm-cubic-cpx-flow-bounded}
  Let $\epsilon$ be smaller or equal to $1/|\beta_{3,2}|$. If $g_{r}\in\cH_{\epsilon}$ then the solution of the Cauchy problem
  \eqref{eq-complex-cubic-ODE-CauchyPb} given in
  \cref{thm-complex-cubic-ODE} is bounded above (in modulus) by $\epsilon$.
\end{cor}
\begin{proof}
  Let us denote by $f$ the function from $\R_{+}$ to $\R_{+}$ defined
  by $f(t)=\envert{1+\alpha t}$. By \cref{thm-complex-cubic-ODE},
  \begin{equation*}
    |g(t)|=\frac{|g_{r}|}{|1+\alpha t+\beta_{3,2} g_{r}\log(1+\alpha
      t)+\beta_{3,2} g_{r}\phi(t)|}
  \end{equation*}
  where $\beta_{3,2}=\beta_{3}/\beta_{2}$. But
  \begin{align}
    \envert{1+\alpha t+\beta_{3,2} g_{r}\log(1+\alpha
      t)+\beta_{3,2} g_{r}\phi(t)}&\ges \envert{1+\alpha
                                    t}-|\beta_{3,2} g_{r}|\del{|\log(1+\alpha
                                    t)|+|\phi(t)|}\nonumber\\
    &\ges \envert{1+\alpha t}-|\beta_{3,2}
      g_{r}|\del{\envert{\log|1+\alpha t|}+3\pi}.\label{eq-borne-inf-module-denom}
  \end{align}
  Let us assume for a moment that this last lower bound is
  non-negative. Then,
  \begin{align}
    |g(t)|&\les \frac{|g_{r}|}{\envert{1+\alpha t}-|\beta_{3,2}
            g_{r}|\del{\envert{\log|1+\alpha t|}+3\pi}}\les \epsilon\label{eq-borne-sup-g-cubic}\\
    \intertext{\ifft}
    |g_{r}|&\les\frac{\epsilon|1+\alpha t|}{1+\epsilon|\beta_{3,2}|\del{\envert{\log|1+\alpha t|}+3\pi}}.\label{eq-borne-sup-g0}
  \end{align}
  Now, the right-hand side of \cref{eq-borne-inf-module-denom} is
  non-negative \ifft
  \begin{equation}
    |g_{r}|\les\frac{|1+\alpha
      t|}{|\beta_{3,2}|\del{\envert{\log|1+\alpha t|}+3\pi}}.\label{eq-condition-g0-denom-positif}
  \end{equation}
  Consequently, if we impose \eqref{eq-borne-sup-g0},
  \cref{eq-condition-g0-denom-positif} is satisfied and this justifies \aposteriori
  the first inequality of \cref{eq-borne-sup-g-cubic}.
  
  Finally, let $g$ be the function from $\R_{+}$ into itself defined by
  $g(x)=\frac{\epsilon x}{1+\epsilon|\beta_{3,2}|\del{|\log x|+3\pi}}$. It is easy to
  check that if $\epsilon\les 1/|\beta_{3,2}|$, $g$ is increasing. As a
  consequence, from \cref{eq-min-f-of-t},
  \begin{equation*}
    |g_{r}|\les
    \begin{cases}
      \frac{\epsilon}{1+3\pi|\beta_{3,2}|\epsilon}&\text{if
        $\theta_{r}\in\intcc[0]{0,\frac\pi 2}$,}\\
      \frac{\epsilon|\sin\theta_{0}|}{1+\epsilon|\beta_{3,2}|\del{\envert{\log|\sin\theta_{r}|}+3\pi}}&\text{if
      $\theta_{r}\in\intco[0]{\frac\pi 2,\pi}$}
    \end{cases}
  \end{equation*}
  implies \eqref{eq-borne-sup-g0} and $|g(t)|\les \epsilon$.
\end{proof}

\subsection{Differential flow of higher degree}
\label{sec-diff-flow-higher-degree}

Let us now consider more general complex differential equations and
prove that for sufficiently small initial conditions, their solutions
are uniformly bounded as well. Let $U$ be a complex neighbourhood of
$0$. Let $f$ be the following function:
\begin{equation}\label{eq-higher-vector-field}
  \begin{aligned}
    f:\R_{+}\times U&\to\C\\
    (t,z)&\mapsto \beta_{2}z^{2}+\beta_{3}z^{3}+z^{4}h(z)
  \end{aligned}
\end{equation}
where $h$ is holomorphic on $U$. We are interested in the following
Cauchy problem:
\begin{subequations}
  \label{eq-continuous-higher-cpx-Cauchy-pb}
    \begin{align}
    g'&=f(t,g)\label{eq-continuous-higher-cpx-ODE}\\
    g(0)&=g_{r}\in\C.\label{eq-continuous-higher-cpx-initial-cond}
  \end{align}
\end{subequations}

\begin{defn}[Disks]
  \label{def-disks}
  Let $r$ be real and positive. We will denote by $\bbD_{r}$ the open
  disk of radius $r$ centered at $0$. An open disk $\bbS_{r}$ of
  radius $r$ centered at $r$ will be called a \NSd. $\bbS_{r}$ is the
  set of complex numbers $z$ such that $\Re\del{\frac 1z}>\frac 1{2r}$
  or equivalently $|z|<2r\cos(\arg z)$.
\end{defn}

\begin{thm}\label{thm-complex-higher-ODE-upper-bound}
  Let $\epsilon$ be a sufficiently small positive real number. There
  exists a simply connected domain $D_{\epsilon}$ of $\C$ such that
  $D_{\epsilon}\subset U$, $0\in\partial D_{\epsilon}$, and
  $D_{\epsilon}$ contains a
  Nevanlinna-Sokal disk $\bbS_{\delta}$, $\delta=\frac
  16\frac{\epsilon}{1+\frac{3\pi}{2}|\beta_{3,2}|\epsilon}$, such that if
  $g_{r}\in D_{\epsilon}$ then, for all $t\ges 0$, the unique
  maximal solution of the Cauchy problem
  \eqref{eq-continuous-higher-cpx-Cauchy-pb} defined on $\R_{+}$
  belongs to $\bbD_{\epsilon}$.
\end{thm}
\begin{proof}
  By the Leau-Fatou flower theorem, there exists an attracting petal
  directed along the positive real axis, see \cref{fig-petal-2}. Consequently
  we know that there exists an open connected and simply connected
  complex subset as claimed in
  \cref{thm-complex-higher-ODE-upper-bound}. But we have no control on
  the size of the Nevanlinna-Sokal disk it contains. Thus we follow a
  more pedestrian road.\\
  
  By a biholomorphic change of variable $y=\varphi(g)$,
  \cref{eq-continuous-higher-cpx-ODE} rewrites as
  \begin{equation*}
    y'=-y^{2}+\frac{\beta_{3}}{\beta_{2}^{2}}y^{3}.
  \end{equation*}
  This comes from a theorem of G.~Szekeres, the proof of which can be
  found in \cite{Loray2005aa}. Let us briefly repeat the arguments
  here. Let $U$ be a complex neighbourhood of $0$ and $f$ a holomorphic
  function on $U$. The corresponding ODE writes
  \begin{equation*}
    z'=f(z)\iff\frac{\diff*z}{f(z)}=\diff*t.
  \end{equation*}
  The meromorphic differential form $\omega=\frac{\diff*z}{f(z)}$ is
  the dual form of the vector field $f(z)\partial_{z}$. Let us assume
  that $0$ is a pole of $\omega$ of order $p+1$, $p\ges 1$, and note
  the corresponding residue $a_{-1}$. In the case of the vector field
  \eqref{eq-higher-vector-field}, $p=1$ and
  $a_{-1}=-\frac{\beta_{3}}{\beta_{2}^{2}}$.

  Let us explain how to
  prove that
  \begin{equation*}
\omega=\frac{\diff*y}{-y^{p+1}-a_{-1}y^{2p+1}}
\end{equation*}
after a
  biholomorphic change of coordinates $y=\varphi(z)$. Firstly, as $1/f$ has a
  pole of order $p+1$ at $0$, there exists a holomorphic function $v$
  (near $0$) such that
  \begin{equation*}
    \omega=\frac{a_{-1}\diff z}{z}+\diff*{\del{\frac{v}{z^{p}}}},\quad
    v(0)\neq 0.
  \end{equation*}
  Secondly, by integrating the equality
  \begin{equation*}
    \frac{a_{-1}\diff z}{z}+\diff*{\del{\frac{v}{z^{p}}}}=\frac{\diff*y}{-y^{p+1}-a_{-1}y^{2p+1}},
  \end{equation*}
  one obtains an implicit equation for $y$:
  \begin{equation}
    a_{-1}\log z+\frac{v(z)}{z^{p}}=\frac 1{py^{p}}+a_{-1}\log y-\frac{a_{-1}}{p}\log(1+a_{-1}y^{p}).\label{eq-implicit-for-y}
  \end{equation}
  Thirdly, defining $y=\varphi(z)\fide z u(z)$, \cref{eq-implicit-for-y}
  becomes
  \begin{equation*}
    \frac{1}{pu^{p}}+a_{-1}z^{p}\log u-\frac{a_{-1}}{p}z^{p}\log\del{1+a_{-1}(zu)^{p}}-v(z)=0,
  \end{equation*}
  and by the implicit function theorem, there exists a neighbourhood $V$
  of $0$ on which $u$ (then $\varphi$) is holomorphic ($\varphi$ is even
  biholomorphic because $\varphi'(0)=u(0)\neq 0$).\\

  From \cref{thm-complex-cubic-ODE}, there exists a unique maximal
  solution $y(t)$ of the Cauchy problem
  \begin{equation*}
    y'=-y^{2}+\frac{\beta_{3}}{\beta_{2}^{2}}y^{3},\quad y(0)=y_{0}
  \end{equation*}
  which is defined on $\R_{+}$. Moreover, from \cref{thm-cubic-cpx-flow-bounded}, if $r'$ is smaller than
  $\frac{\beta_{2}^{2}}{|\beta_{3}|}$, then if $y_{0}\in\cH_{r'}$,
  $y(t)\in\bbD_{r'}$ for all $t\ges 0$. If $r'$ is small enough then
  $\bar\bbD_{r'}$ and thus $\bar\cH_{r'}$ are subsets of $V$. As a
  consequence, if $g_{r}\in\varphi^{-1}(\cH_{r'})\fide\Omega$ then
  $g(t)=\varphi^{-1}\del{y(t)}$ is the unique maximal solution of
  \eqref{eq-continuous-higher-cpx-Cauchy-pb} defined on $\R_{+}$ and for all $t\ges
  0$, $g(t)\in\varphi^{-1}(\bbD_{r'})$. 

  Note that $\varphi^{-1}(\bbD_{r'})$ is bounded, open and connected as
  the image of a bounded, open and (arc-)connected subset by a
  (non-constant) holomorphic function. So let us prove that $\varphi^{-1}(\cH_{r'})$
  contains a \NSd by showing that there exists $\delta >0$ such that
  $\varphi(\bbS_{\delta})\subset\cH_{r'}$. We proceed in two steps. We first prove that if $z\in\bbS_{\delta}$ then $\Re\varphi(z)\ges
  0$. Indeed, there
  exists a holomorphic function $\chi$ on $V$ such that
  $\varphi(z)=-\beta_{2}z+z^{2}\chi(z)$. Let $\theta$ be an argument
  of $z$.
  \begin{align*}
    \Re\varphi(z)&\ges 0\iff |\beta_{2}||z|\cos(\theta) \ges |z|^{2}\cos(2\theta)\Re\chi(z)-|z|^{2}\sin(2\theta)\Im\chi(z).
  \end{align*}
  If $|z|=0$ then $\Re\varphi(z)=0$. Otherwise,
  \begin{align*}
    \Re\varphi(z)&\ges 0\iff |\beta_{2}|\cos(\theta) \ges |z|\cos(2\theta)\Re\chi(z)-|z|\sin(2\theta)\Im\chi(z).
  \end{align*}
  But
  \begin{align*}
    |z|\cos(2\theta)\Re\chi(z)-|z|\sin(2\theta)\Im\chi(z)&\les
                                                           2\delta K_{\delta}\cos(\theta)\del{|\cos(2\theta)|+|\sin(2\theta)|}\\
    &\les 4\delta K_{\delta}\cos(\theta)
  \end{align*}
where $K_{\delta}=\sup_{z\in\bar\bbS_{\delta}}|\chi(z)|$ and $\lim_{\delta\to
  0}K_{\delta}=0$.  Thus, for $\delta$ small enough, $4 \delta K_{\delta}<|\beta_{2}|$ and
$\Re\varphi(z)\ges 0$.

We then show that
$|\varphi(z)|\les\frac{r'}{1+3\pi\frac{|\beta_{3}|}{\beta_{2}^{2}}r'}$:
\begin{align*}
  |\varphi(z)|\les|\beta_{2}||z|+K_{\delta}|z|^{2}&\les
                                                    2|\beta_{2}|\delta\del[2]{1+2\frac{\delta
                                                    K_{\delta}}{|\beta_{2}|}}\les 3|\beta_{2}|\delta.
\end{align*}
As a consequence, fixing
\begin{equation*}
  3|\beta_{2}|\delta=\frac{r'}{1+3\pi\frac{|\beta_{3}|}{\beta_{2}^{2}}r'},\quad \epsilon=\frac{2}{|\beta_{2}|}r',
\end{equation*}
the theorem is proved.
\end{proof}

If the initial value $g_{r}$ is real and $h$ real-valued, we can be more precise:
\begin{thm}
  There exists $g_{c}\in\R_{+}^{*}$
  such that for all $g_{r}$ real in $(0,g_{c})$, the Cauchy problem
 \eqref{eq-continuous-higher-cpx-Cauchy-pb} has a unique
  decreasing solution $g$ defined on $\R_{+}$. Moreover let $\epsilon$
  be a positive real number smaller than $1$. Then there exists a
  positive real number $\alpha(\epsilon)$ (smaller than $1$) such that if
  $g_{r}\in(0,\alpha g_{c})$, $g$ satisfies
  \begin{multline}
    \label{eq-continuous-higher-flow-solution-estimates}
    \frac{g_{r}}{1-\beta_{2}g_{r}t+\frac{\beta^{-}_{3}}{\beta_{2}}g_{r}\log(1-\beta_{2}g_{r}t)+\frac{\beta^{-}_{3}}{\beta_{2}}g_{r}\phi_{-}(t)}<g(t)\\
    <\frac{g_{r}}{1-\beta_{2}g_{r}t+\frac{\beta^{+}_{3}}{\beta_{2}}g_{r}\log(1-\beta_{2}g_{r}t)+\frac{\beta_{3}^{+}}{\beta_{2}}g_{r}\phi_{+}(t)},
  \end{multline}
  with $\beta_{3}^{-}\defi
  (1-\sgn(\beta_{3})\epsilon)\beta_{3}$, $\beta_{3}^{+}\defi
  (1+\sgn(\beta_{3})\epsilon)\beta_{3}$ and $\phi_{-},\phi_{+}$ two
  bounded functions on $\R_{+}$.
\end{thm}
\begin{proof}
  By the Cauchy-Lipschitz theorem, for all $g_{r}\in U\cap\R$, there exists
  a unique solution of the Cauchy problem
 \eqref{eq-continuous-higher-cpx-Cauchy-pb} defined on $[0,T)$ for some
  $T>0$. Let $a:U\cap\R\to\R$ be defined as
  \begin{equation*}
    f(t,x)\defi\beta_{2}x^{2}a(x),
  \end{equation*}
  so that $a(x)=1+\beta_{3,2} x+\tfrac 1{\beta_{2}}x^{2}h(x)$. Let $g_{c}$
  be the smallest positive zero of $a$ in $U\cap\R$ if it exists and
  $\sup U\cap\R$ otherwise. As $f(t,x)$ is negative if $x\in (0,g_{c})$, by 
  unicity of the solutions of the Cauchy problem
 \eqref{eq-continuous-higher-cpx-Cauchy-pb}, for all $t\in [0,T)$, $f(t,g(t))$
  is negative. As a consequence, $0<g(t)<g_{r}$ and $g$ is in fact
  defined on $\R_{+}$ (and decreasing).\\
  
  Moreover if $g_{r}$ is small enough (say smaller than $\alpha g_{c}$),
  \begin{equation*}
    \beta_{2}g^{2}+(1-\sgn(\beta_{3})\epsilon)\beta_{3}g^{3}<\beta_{2}g^{2}+\beta_{3}g^{3}+g^{4}h(g)<\beta_{2}g^{2}+(1+\sgn(\beta_{3})\epsilon)\beta_{3}g^{3}
  \end{equation*}
  for all $t\ges 0$ and by \cref{thm-complex-cubic-ODE}, $g$ satisfies \cref{eq-continuous-higher-flow-solution-estimates}.
\end{proof}

\newpage
\printbibliography

\contactrule
\contactVRivasseau
\contactFVT

\end{document}
